\include{BoxedEPS}

\documentclass[11pt]{amsart}
\usepackage{amsmath,amsthm}
\usepackage{graphicx,amsmath,amsfonts}

\chardef\bslash=`\\ 

\makeatletter
\def\verbatim{\interlinepenalty\@M \@verbatim
  \leftskip\@totalleftmargin\advance\leftskip2pc
  \frenchspacing\@vobeyspaces \@xverbatim}
\makeatother
\newtheorem{thm}{Theorem}[section]
\newtheorem{cor}[thm]{Corollary}

\newtheorem{prop}[thm]{Proposition}
\newtheorem{ex}[thm]{Example}
\newtheorem{rem}[thm]{Remark} 

\numberwithin{equation}{section}

\newcommand{\ZZ}{{\mathbb Z}}
\newcommand{\RR}{{\mathbb R}}

\begin{document}


\title{Nonlinear frames and sparse reconstructions in Banach spaces}
\author{Qiyu Sun 
and Wai-Shing Tang}

\thanks{The project is  partially supported by
 the National Science Foundation (DMS-1412413) and
 Singapore Ministry of Education Academic Research Fund Tier 1 Grant (No. R-146-003-193-112).
 }

\address{Qiyu Sun: Department of Mathematics,  University of Central Florida,
Orlando, FL 32816, USA. Email: qiyu.sun@ucf.edu}

\address{Wai-shing Tang, Department of Mathematics, National University of Singapore,  Singapore 119076
 Republic of Singapore. Email: mattws@nus.edu.sg}


\date{\today }


\maketitle

\begin{abstract}
In the first part of this paper, we consider nonlinear extension of frame theory
by introducing bi-Lipschitz maps $F$ between Banach spaces.
Our linear model of bi-Lipschitz maps  is the analysis operator
associated with Hilbert frames, $p$-frames, Banach frames, g-frames and fusion
frames. In general Banach space setting, stable algorithm to reconstruct
a signal $x$ from its noisy measurement $F(x)+\epsilon$
 may not exist.  
In this paper, we establish exponential convergence of two iterative reconstruction algorithms
when $F$ is not too far from some bounded below linear operator with bounded pseudo-inverse, and
when $F$ is a well-localized map between two Banach spaces with dense Hilbert subspaces.
The crucial step to prove the later conclusion is a novel fixed point theorem   for a
well-localized map on a Banach space.

In the second part of this paper, we consider
 stable reconstruction of sparse signals
in a union ${\bf A}$ of closed linear subspaces of a Hilbert space ${\bf H}$ from their nonlinear measurements.
We create an optimization framework called  sparse
approximation triple $({\bf A}, {\bf M}, {\bf H})$, and show that
the minimizer
$$x^*={\rm argmin}_{\hat x\in {\mathbf M}\ {\rm with} \  \|F(\hat x)-F(x^0)\|\le \epsilon} \|\hat x\|_{\mathbf M}$$
 provides a suboptimal approximation to the original sparse signal $x^0\in {\bf A}$
when the  measurement map $F$
 has the
  sparse Riesz property
  and  almost linear property on ${\mathbf A}$.
The above two  new properties is  also discussed in this paper when $F$ is not far away
from a linear measurement operator $T$ having the restricted isometry property.
\end{abstract}



\section{Introduction}

For a Banach  space ${\mathbf B}$, we denote its norm by $\|\cdot\|_{\mathbf B}$. 
A  map $F$ from one Banach space ${\mathbf B}_1$ to
another Banach space ${\mathbf B}_2$  is said to have
 {\em bi-Lipschitz property}  if
there exist two positive constants $A$ and $B$ such that
\begin{equation}\label{bilipschitzmap.def}
A\|x-y\|_{{\mathbf B}_1}\le \|F(x)-F(y)\|_{{\mathbf B}_2}\le B \|x-y\|_{{\mathbf B}_1}\quad {\rm for \ all} \ x, y\in {\mathbf B}_1.
\end{equation}
Our models of bi-Lipschitz maps between Banach spaces  are analysis operators associated with Hilbert frames, $p$-frames,  Banach frames, $g$-frames and fusion frames
 \cite{ast01,  chl99, ckl08,  Christensen03,  wsun06}.
  Our study is also motivated by
 nonlinear sampling theory and phase retrieval,
which have gained  substantial attention in recent
years
 \cite{balan09, bce06, bandeira.acha14, csvappear, dem08, eldar.acha14, kst95, sunaicm13}.
  The framework developed in the first part of  this paper could be considered
 as a nonlinear extension of frame theory. 

\smallskip
 Denote by
 ${\mathcal B}({\mathbf B}_1, {\mathbf B}_2)$
 the Banach space of all bounded linear operators from one Banach space ${\mathbf B}_1$ to another Banach space ${\mathbf B}_2$.
A continuous map $F$ from ${\mathbf B}_1$ to ${\mathbf B}_2$
is said to
be  
 {\em differentiable at $x\in {\mathbf B}_1$} if there exists a  linear operator,
denoted by $F'(x)$, in ${\mathcal B}({\mathbf B}_1, {\mathbf B}_2)$
 such that
\begin{equation*}
\lim_{y\to 0} \frac{\|F(x+y)-F(x)-F'(x)y\|_{{\mathbf B}_2}}{\|y\|_{{\mathbf B}_1}}=0;
\end{equation*}
 and to be {\em differentiable  on ${\mathbf B}_1$} if it is differentiable at every $x\in {\mathbf B}_1$ \cite{dieu69}.
For a differentiable bi-Lipschitz map $F$
from  ${\mathbf B}_1$ to  ${\mathbf B}_2$,
one may easily verify that its derivatives $F'(x),x\in {\mathbf B}_1$, are {\em uniformly stable}, i.e.,
 there exist two positive constants $A$ and $B$ such that
\begin{equation}\label{localstability.thm.eq1}
A\|y\|_{{\mathbf B}_1}\le \|F'(x)y\|_{{\mathbf B}_2}\le B \|y\|_{{\mathbf B}_1}\quad {\rm for \ all} \ x, y\in {\mathbf B}_1.
\end{equation}
The converse is not true in general. 
Then we have the following natural question. 

\noindent {\bf Question 1}: \ {\em When does a differentiable map with
the uniform stability property \eqref{localstability.thm.eq1} 
have the bi-Lipschitz property \eqref{bilipschitzmap.def}?}

We say that a  linear operator
 $T\in {\mathcal B}({\mathbf B}_1, {\mathbf B}_2)$ from  one Banach space ${\mathbf B}_1$ to another Banach space ${\mathbf B}_2$
  is {\em  bounded below} if
 \begin{equation}\label{Tstable.condition}
 \inf_{0\ne y\in {\mathbf B}_1} \frac
 {\|Ty\|_{{\mathbf B}_2}}{\|y\|_{{\mathbf B}_1}}>0.\end{equation}
For a continuously differentiable map $F$  not  too nonlinear, particularly
 not far away from a bounded below linear operator $T$,  a sufficient  condition
 for \eqref{bilipschitzmap.def}  is that
 for any $0\ne y\in {\mathbf B}_1$, the set ${\mathbb B}(y)$ of
unit vectors $F'(x)y/\|F'(x)y\|_{{\mathbf B}_2}, x\in {\mathbf B}_1$,
is contained in  a ball of radius
\begin{equation}\label{betaft.condition}
\beta_{F,T}<1\end{equation}
with center at  ${Ty}/{\|Ty\|_{{\mathbf B}_2}}$, where 
\begin{equation}\label{linearapproximation.center}
\beta_{F, T} := \sup_{0\ne y\in {\mathbf B}_1} \sup_{x\in {\mathbf B}_1}
\Big\|\frac{F'(x)y}{\|F'(x)y\|_{{\mathbf B}_2}}-\frac{Ty}{\|Ty\|_{{\mathbf B}_2}}\Big\|_{{\mathbf B}_2}.
\end{equation}
The above geometric requirement on  the radius
 $\beta_{F,T}$
  is optimal in  Banach space setting, but it  could be relaxed to
 \begin{equation} \label{frame.sufficientcondition.hilbert.eq1} \beta_{F,T}<\sqrt{2}\end{equation}
  in  Hilbert space setting,
which implies that for any  $0\ne y\in {\mathbf B}_1$, the set ${\mathbb B}(y)$ is contained in
 a right circular cone  with axis $Ty/\|Ty\|_{{\mathbf B}_2}$ and  angle strictly less than $\pi/2$.
 Detailed arguments of the above conclusions on a differentiable map
 are given in Appendix \ref{bilipschitz.appendix}.

\smallskip

 Denote by $F({\mathbf B}_1)\subset {\mathbf B}_2$
 the image of a map $F$ from  one Banach space ${\mathbf B}_1$ to another Banach space ${\mathbf B}_2$.
For a bi-Lipschitz map $F: {\mathbf B}_1\to {\mathbf B}_2$, as it is one-to-one, for any $y\in F({\mathbf B}_1)$ there exists a unique $x\in {\mathbf B}_1$  such that $F(x)=y$.
 Our next question is as follows:

 \noindent {\bf Question 2}:\ {\em Given  noisy observation  $z_\epsilon=F(x^0)+\epsilon$
 of $x^0\in {\mathbf B}_1$  corrupted by 
  $\epsilon\in {\mathbf B}_2$, 
    how to construct a suboptimal approximation $x\in {\mathbf B}_1$ 
     such that
   \begin{equation}\label{suboptimal.problem}
   \|x-x^0\|_{{\mathbf B}_1}\le C \|\epsilon\|_{{\mathbf B}_2},
   \end{equation}
   where $C$ is an absolute constant independent of  $x^0\in {\mathbf B}_1$ and $\epsilon\in {\mathbf B}_2$?}

  For a differentiable bi-Lipschitz map $F$
  not far away from a bounded below linear operator $T$,
define  $x_n, n\ge 0$, iteratively with arbitrary initial $x_0\in {\mathbf B}_1$ by 
 \begin{equation} \label{banachalgorithm1.intro}
 x_{n+1}=x_n-\mu T^\dag (F(x_n)-z_\epsilon),\ n\ge 0,\end{equation}
 where
  $T^\dag$ is a bounded  left-inverse of the linear operator $T$, and
 the relaxation factor $\mu$ satisfies
 $0<\mu\le (\sup_{x\in {\mathbf B}_1}\sup_{y\ne 0} {\|F'(x) y\|_{{\mathbf B}_2}}/{\|Ty\|_{{\mathbf B}_2}})^{-1}$.
 In Theorem \ref{banachalgorithm.thm} of
 Section \ref{algorithm.section},
we show 
 that the sequence $x_n, n\ge 0$, in the iterative algorithm \eqref{banachalgorithm1.intro}
converges exponentially to a suboptimal approximation element  $x\in {\mathbf B}_1$ satisfying \eqref{suboptimal.problem},   provided that
\begin{equation}\label{banachalgorithm.thm.eq1} \beta_{F,T}< (\|T\|_{{\mathcal B}({\mathbf B}_1,
{\mathbf B}_2)} \|T^\dag\|_{{\mathcal B}({\mathbf B}_2,
{\mathbf B}_1)})^{-1}.\end{equation}
The  above requirement \eqref{banachalgorithm.thm.eq1} about $\beta_{F,T}$ to guarantee  convergence of the iterative algorithm
\eqref{banachalgorithm1.intro}
is stronger than the sufficient condition
\eqref{betaft.condition}
 for the bi-Lipschitz property of the map $F$.
In Theorem \ref{hilbertalgorithm.thm}, we close that requirement gap  on  $\beta_{F,T}$ in Hilbert space setting
by introducing an iterative algorithm of Van-Cittert type,
 \begin{equation} \label{vancittert.algorithm.intro}
 u_{n+1}=u_n-\mu T^* (F(u_n)-z_\epsilon),\ n\ge 0,\end{equation}
 where
  $T^*$ is the conjugate of the linear operator $T$
   and $\mu>0$ is a 
  small
  relaxation factor.

\smallskip

In the iterative algorithm \eqref{banachalgorithm1.intro},
a left-inverse $T^\dag$ of the bounded below linear operator $T$   is used, but its existence
is not always assured in Banach space setting and its construction is not necessarily attainable even it exists.
This limits applicability of the iterative reconstruction algorithm \eqref{banachalgorithm1.intro}.
In fact,  for general Banach space setting, a stable reconstruction algorithm 
 may not exist
 \cite{chl99, christensenstoeva03}.
 On the other hand, a stable iterative algorithm is proposed
 in \cite{sunaicm13} to find sub-optimal approximation
  for well-localized nonlinear maps on sequence spaces $\ell^p(\ZZ), 2\le p\le \infty$.
 So we have the following question. 

\noindent {\bf Question 3}:\ {\em
  For what types of Banach spaces ${\mathbf B}_1$ and ${\mathbf B}_2$ and nonlinear maps $F$
  from ${\mathbf B}_1$ to ${\mathbf B}_2$
 does  there exist  a stable
   reconstruction  of $x^0\in {\mathbf B}_1$ from its nonlinear observation $y=F(x^0)\in {\mathbf B}_2$?
   }

We say that  a Banach space ${\mathbf B}$ with norm $\|\cdot\|_{\mathbf B}$
is {\em Hilbert-dense} (respectively {\em weak-Hilbert-dense})  if there exists a Hilbert subspace ${\mathbf H}\subset {\mathbf B}$
with norm $\|\cdot\|_{\mathbf H}$
such that ${\mathbf H}$ is dense in ${\mathbf B}$ in the strong topology
(respectively in the  weak topology) of ${\mathbf B}$ and
$$\sup_{0\ne x\in {\mathbf H}} \frac{\|x\|_{\mathbf B}}{\|x\|_{\mathbf H}}<\infty.$$
Our models of the above new concepts are the sequence spaces $\ell^p, 2\le p\le \infty$, for which $\ell^p$ with  $2\le p<\infty$ are Hilbert-dense and $\ell^\infty$ is weak-Hilbert-dense. 
 For  (weak-)Hilbert-dense Banach spaces ${\mathbf B}_1$ and ${\mathbf B}_2$ and a nonlinear map $F:{\mathbf B}_1\to {\mathbf B}_2$
that has certain localization property,  a stable
   reconstruction  of $x^0\in {\mathbf B}_1$ from its nonlinear observation $y=F(x^0)\in {\mathbf B}_2$
   is proposed in  Theorem \ref{banachwiener.algorithm.thm} of Section \ref{localizedmaps.section}.
   The crucial step 
is a new fixed point theorem for a  well-localized differentiable map  whose restriction on a dense Hilbert subspace 
is a contraction, see Theorem \ref{fixedpoint.thm}.

\bigskip

Let ${\mathbf H}_1$ and ${\mathbf H}_2$ be Hilbert spaces, and
let ${\mathbf A}=\cup_{i\in I} {\mathbf A}_i$
 be union of  closed linear subspaces ${\mathbf A}_i, i\in I$, of  the  Hilbert space ${\mathbf H}_1$.
 The second topic of this paper is to study
the {\em restricted bi-Lipschitz property} of a map $F: {\mathbf H}_1\to {\mathbf H}_2$
 on
${\mathbf A}$, which means that
 there exist two positive constants $A$ and $B$ such that
\begin{equation}\label{sparsebilipschitzmap.def}
A\|x-y\|_{{\mathbf H}_1}\le \|F(x)-F(y)\|_{{\mathbf H}_2}\le B \|x-y\|_{{\mathbf H}_1}\quad {\rm for \ all} \ x, y\in {\mathbf A}.
\end{equation}
This topic is motivated by sparse recovery problems on finite-dimensional spaces
\cite{crt06, ct05, eldarbook2012, foucartbook}. 
As we use the union ${\bf A}$
 of  closed linear spaces ${\mathbf A}_i, i\in I$,
to model sparse signals, the restricted bi-Lipschitz property
 of a map $F$ could be thought as nonlinear correspondence of restricted isometric property of a measurement matrix.
So the framework developed in  the second part
is nonlinear Banach space extension of the  finite-dimensional  sparse recovery problems.

\smallskip

  In the classical sparse recovery  setting \cite{crt06, ct05, eldarbook2012, foucartbook}, 
 the set of all $s$-sparse signals for some $s\ge 1$ is used as the set ${\mathbf A}$. In this case,
 elements in ${\mathbf A}$ can be described by their $\ell^0$-quasi-norms being less than or equal to $s$,
 and the sparse recovery problem could  reduce to the $\ell^0$-minimization problem. Due to
  numerical infeasibility of the $\ell^0$-minimization, a relaxation to  (non-)convex $\ell^q$-minimization with $0<q\le 1$
 was proposed, and more importantly it was proved that the $\ell^q$-minimization recovers sparse signals
when the linear measurement operator has certain restricted isometry property in $\ell^2$ \cite{crt06, ct05, 
foucart.acha09, foucartbook,
sunacha12}. This leads to the following question.

\noindent{\bf Question 4}: {\em How to create a general optimization framework  to recover  sparse signals?}

Given a Banach space ${\mathbf M}$, 
we say that a subset $K$ of ${\bf M}$ is
{\em proximinal} (\cite{borwein89, lau79}) if  every element  $x\in {\bf M}$ has a best approximator $y\in K$, that is,
$$\|x-y\|_{\mathbf M}=\inf_{z\in K} \|x-z\|_{\mathbf M}=:\sigma_{K, {\bf M}}(x).$$
 Given Hilbert spaces ${\mathbf H}_1$ and ${\mathbf H}_2$,
 a union
${\mathbf A}=\cup_{i\in I} {\mathbf A}_i$  of  closed linear subspaces ${\mathbf A}_i, i\in I$, of ${\mathbf H}_1$,
and a continuous map $F$ from ${\mathbf H}_1$ to ${\mathbf H}_2$, consider the following minimization problem in a Banach space ${\mathbf M}$, 
\begin{equation}\label{minimization.intro}
x^*={\rm argmin}_{\hat x\in {\mathbf M} \ {\rm with} \ F(\hat x)=z} \|\hat x\|_{\mathbf M}
\end{equation}
for any given observation $z:=F(x)$ for some $x\in {\mathbf A}$.
To make the above minimization problem suitable for stable reconstruction of $x\in {\mathbf A}$ from its observation $F(x)$, we  introduce the concept of a
 {\em sparse approximation triple}
$({\mathbf A}, {\mathbf M}, {\mathbf H}_1)$:

 \begin{itemize}

 \item[{(i)}] (Continuous imbedding property) \ The Banach space ${\mathbf M}$  contains
all elements in ${\mathbf A}$ and  it is contained in the Hilbert space ${\mathbf H}_1$, that is,
\begin{equation}\label{ambimbedding.intro}
{\mathbf A}\subset {\mathbf M}\subset {\mathbf H}_1,\end{equation}
and the  imbedding operators $i_{\mathbf A}: {\mathbf A}\to {\mathbf M}$ and
 $i_{\mathbf M}: {\mathbf M}\to {\mathbf H}_1$
     are bounded.

  \item[{(ii)}] (Proximinality property) The Banach space ${\mathbf M}$
  has ${\bf A}$ as its closed subset, and  all closed  subsets of ${\bf M}$
   being proximinal.

 \item[{(iii)}] (Common-best-approximator property)\ Given  any $i\in I$, a best approximator
 $x_{{\mathbf A}_i, \mathbf M}:={\rm argmin}_{\hat x\in {\mathbf A}_i}\|\hat x-x\|_{\mathbf M}$
of $x\in {\mathbf M}$ in the norm $\|\cdot\|_{\mathbf M}$ is also a best approximator in
 the norm $\|\cdot\|_{{\mathbf H}_1}$, that is,
\begin{equation} \label{bestapproximation.hm} x_{{\mathbf A}_i, \mathbf M}={\rm argmin}_{\hat x\in {\mathbf A}_i}\|\hat x-x\|_{{\mathbf H}_1}.
\end{equation}

 \item[{(iv)}] (Norm-splitting property) \ For the best approximator
 $x_{{\mathbf A}_i, \mathbf M}$
of $x\in {\mathbf M}$ in the norm $\|\cdot\|_{\mathbf M}$, 
\begin{equation}\label{optimization.thm1.eq3}
\|x\|_{\mathbf M}= \|x_{{\mathbf A}_i, \mathbf M}\|_{\mathbf M}+\|x-x_{{\mathbf A}_i, \mathbf M}\|_{\mathbf M},
\end{equation}
and
\begin{equation}\label{bestapproximation.hm.equiv}
\|x\|_{{\mathbf H}_1}^2= \|x_{{\mathbf A}_i, \mathbf M}\|_{{\mathbf H}_1}^2+\|x-x_{{\mathbf A}_i, \mathbf M}\|_{{\mathbf H}_1}^2.
\end{equation}

\item[{(v)}] (Sparse density property)  $\cup_{k\ge 1} k{\bf A}$ is dense in ${\bf H}_1$, where
 $$k{\mathbf A}:=\underbrace{{\mathbf A}+{\mathbf A}+\cdots+ {\mathbf A}}_{k \  {\rm times}}=\Big\{\sum_{i=1}^k x_i: \ x_1, \ldots, x_k\in {\mathbf A}\Big\},\ k\ge 1. $$

 \end{itemize}

 One may easily verify that these five properties  are satisfied for the triple $({\mathbf A}, {\mathbf M}, {\mathbf H}_1)$
in the classical sparse recovery setting, 
 where ${\mathbf A}$
is the set of all $s$-sparse vectors, ${\mathbf M}$ is the set  of all summable sequences, and
${\mathbf H}_1$ is the set of all square-summable sequences \cite{crt06, ct05, eldarbook2012, foucartbook}.

In this paper, we rescale the norm $\|\cdot\|_{\mathbf M}$ in the sparse approximation triple $({\bf A}, {\bf M}, {\bf H}_1)$  so that
the imbedding operator $i_{\mathbf M}$ has norm one, 
\begin{equation} \label{optimization.thm1.eq1}
\|i_{\bf M}\|_{{\mathcal B}({\bf M}, {\bf H}_1)}= 1,
\end{equation}
otherwise replacing it by $\|\cdot\|_{\bf M} \|i_{\bf M}\|_{{\mathcal B}({\bf M}, {\bf H}_1)}$.
Next we introduce two quantities  to measure
 {\em sparsity}
 \begin{equation} \label{rationB1M}
s_{\mathbf A}:=\|i_{\bf A}\|_{{\mathcal B}({\bf A}, {\bf M})}^2  
=\Big(\sup_{0\ne x\in {\mathbf A}} \frac    {\|x\|_{\mathbf M}} {\|x\|_{{\mathbf H}_1}}\Big)^2
\end{equation}
for signals in  ${\mathbf A}$, and {\em sparse approximation ratio}
\begin{equation} \label{secondapproximationestimate}
 a_{\mathbf A}:= \sup_{0\ne x\in {\mathbf M}}
\Big(\frac{\|u_{\mathbf A, \mathbf M}\|_{{\mathbf H}_1}}{
\|x_{\mathbf A, \mathbf M}\|_{\mathbf M}}\Big)^2\le 1
\end{equation}
for elements in ${\mathbf M}$,
where $x_{{\mathbf A}, {\mathbf M}}$
and $u_{\mathbf A, \mathbf M}\in {\bf A}$ are the first and second best approximators of $x$  respectively,
\begin{equation}\label{Abestapproximator}
\|x-x_{{\mathbf A}, {\mathbf M}}\|_{\mathbf M}= \sigma_{{\bf A}, {\bf M}}(x)\ {\rm and} \
\|x-x_{{\mathbf A}, {\mathbf M}}-u_{{\mathbf A}, {\mathbf M}}\|_{\mathbf M}= \sigma_{{\bf A}, {\bf M}}(x-x_{{\bf A}, {\bf M}}).\end{equation}
The upper bound estimate  in \eqref{secondapproximationestimate} 
holds, since
 \begin{eqnarray*} \|u_{{\bf A}, {\bf M}}\|_{\bf M} & = &  \|x-x_{{\bf A}, {\bf M}}\|_{\bf M}-\|x- x_{{\bf A}, {\bf M}}-u_{{\bf A}, {\bf M}}\|_{\bf M}\\
 & \le  &  \|x-u_{{\bf A}, {\bf M}}\|_{\bf M}-\|x- x_{{\bf A}, {\bf M}}-u_{{\bf A}, {\bf M}}\|_{\bf M}
\le \| x_{{\bf A}, {\bf M}}\|_{\bf M}, \ x\in {\bf M},
\end{eqnarray*}
by the norm splitting property \eqref{optimization.thm1.eq3}.
 In the classical sparse recovery setting with ${\mathbf A}$ being the set of all $s$-sparse signals,
  one may verify that
  $s_{\mathbf A}=s$ and $a_{\mathbf A}=1/s$, see  Appendix \ref{sat.appendix} for additional properties of
sparse approximation triples.

\smallskip
Having introduced the sparse approximation triple
$({\mathbf A}, {\mathbf M}, {\mathbf H}_1)$, our next question is  on optimization approach to sparse signal recovery, see \cite{beck2013, bdieee13, bdieee09, bdieee11, bdieee13,  ehler14, emieee09, ldieee08} for  the classical setting.

\noindent {\bf Question 5}:
{\em Given a sparse approximation triple $({\bf A}, {\bf M}, {\bf H}_1)$, for what type of maps $F: {\mathbf H}_1\to {\mathbf H}_2$ does the solution $x_{\mathbf M}^0$ of  the optimization problem
\begin{equation}\label{optimizationsolution}
x_{\mathbf M}^0:={\rm argmin}_{\hat x\in {\mathbf M}\ {\rm with} \ \|F(\hat x)-F(x^0)\|\le \varepsilon} \|\hat x\|_{\mathbf M},
\end{equation}
is  a suboptimal approximation to  the sparse signal $x^0$ in ${\mathbf A}$?}

In this paper, without loss of generality, we  assume that $F(0)=0$.
We say that  $F: {\mathbf H}_1\to {\mathbf H}_2$  has  {\em sparse Riesz property} if
\begin{equation}\label{sap.def}
 \|F(x)\|_{{\mathbf H}_2}\ge D^{-1} \big(\|x\|_{{\mathbf H}_1}- \beta  \sqrt{a_{\mathbf A}}\
  \sigma_{{\mathbf A}, {\mathbf M}}(x)\big), \
x\in {\mathbf M},
\end{equation}
with  $D, \beta>0$,
and that $F$ is {\em almost linear on ${\bf A}$} if
\begin{equation}\label{almostlinearproperty}
 \|F(x)-F(y)-F(x-y)\|_{{\mathbf H}_2}\le \gamma_1 \|x-y\|_{{\mathbf H}_1}+\gamma_2 \sqrt{a_{\mathbf A}}
 (\sigma_{\mathbf A, \mathbf M}(x)+ \sigma_{\mathbf A, \mathbf M}(y)), \ x, y\in {\mathbf M},
\end{equation}
with $\gamma_1, \gamma_2\ge 0$.    Combining the sparse Riesz property and almost linear property of a map $F$ gives
 \begin{eqnarray}
 \|F(x)-F(y)\|_{{\mathbf H}_2} & \ge &  (D^{-1}-\gamma_1) \|x-y\|_{{\mathbf H}_1}\nonumber\\
 & & - (D^{-1}\beta+\gamma_2) \sqrt{a_{\mathbf A}}
  (\sigma_{\mathbf A, \mathbf M}(x)+ \sigma_{\mathbf A, \mathbf M}(y)), \ x, y\in {\mathbf M},
 \end{eqnarray}
and hence $F$ has the  restricted bi-Lipschitz property on ${\bf A}$,
 \begin{equation*}
 \|F(x)-F(y)\|_{{\mathbf H}_2}  \ge   (D^{-1}-\gamma_1) \|x-y\|_{{\mathbf H}_1}, \ \ x, y\in {\mathbf A},
 \end{equation*}
 when $\gamma_1$ and $D$ satisfy $D\gamma_1<1$.
 In  
 Section \ref{optimization.section}, 
 we show that
the solution $x_{\mathbf M}^0$ of  the optimization problem
\eqref{optimizationsolution}
is  a suboptimal approximation to  the  signal $x^0$ in ${\mathbf M}$, i.e.,
there exist positive constants $C_1$ and $C_2$ such that
\begin{equation}\label{errorestimate.last}
\|x_{\bf M}^0-x^0\|_{{\mathbf H}_1}\le C_1 \sqrt{a_{\bf A}} \sigma_{{\bf A}, {\bf M}}(x^0)+ C_2 \epsilon,
\end{equation}
provided that   $F$ has the sparse Riesz property \eqref{sap.def} and  almost linear  property \eqref{almostlinearproperty}
with
  $D, \beta, \gamma_1$ and $\gamma_2$ satisfying
$$ 1-2D\gamma_1-(D\gamma_1+D\gamma_2+\beta)\sqrt{a_{\mathbf A}s_{\mathbf A}}>0.
$$
We remark that the approximation error estimate \eqref{errorestimate.last} implies the sparse Riesz property
 \eqref{sap.def} for the map $F$,
$$\|F(x)\|_{{\mathbf H}_2}\ge C_2^{-1} \big(\|x\|_{{\mathbf H}_1}- C_1  \sqrt{a_{\bf A}} \sigma_{{\bf A}, {\bf M}}(x)\big),\  x\in {\bf M},$$
which follows from \eqref{errorestimate.last} by taking $x^0=x$ and $\epsilon=\|F(x^0)\|_{{\bf H}_2}$.

The sparse Riesz property 
 was introduced in \cite{sunieee11} with a different name, sparse approximation property,  for
   the classical sparse recovery setting; and the almost linear property was studied in \cite{gevirtz82, junpark96} for bi-Lipschitz maps between Banach spaces.
 In Section \ref{sparserieszproperty.section}, we consider the following  question.

\noindent {\bf Question 6}: {\em When does a map
  $F: {\mathbf H}_1\to {\mathbf H}_2$ have the sparse Riesz property \eqref{sap.def}
  and the almost linear property \eqref{almostlinearproperty}?
 }

We say that
 a linear operator $T\in {\mathcal B}({\mathbf H}_1, {\mathbf H}_2)$ has the {\em restricted isometry property} (RIP)
  on $2{\mathbf A}$ if
 \begin{equation} \label{sparsealgorithm.thm.eq2}
(1-\delta_{2\mathbf A}(T)) \|z\|_{{\mathbf H}_1}^2\le  \|Tz\|_{{\mathbf H}_2}^2\le (1+ \delta_{2\mathbf A}(T))\|z\|_{{\mathbf H}_1}^2\quad {\rm for \ all}\ z\in 2{\mathbf A},
 \end{equation}
where $\delta_{2\mathbf A}(T)\in [0, 1)$ \cite{crt06, ct05}.
A nonlinear map $F: {\mathbf H}_1\longmapsto {\mathbf H}_2$, not far away from a linear operator $T$
with the  restricted isometry property \eqref{sparsealgorithm.thm.eq2} in the sense that
$$
\gamma_{F,T}(2{\bf A})<\sqrt{1-\delta_{2\mathbf A}(T)} $$
 has the restricted bi-Lipschitz property \eqref{sparsebilipschitzmap.def} on ${\mathbf A}$,
 where
 \begin{equation}\label{tildegammafta}
\gamma_{F,T}(k{\bf A}):=\sup_{x\in {\mathbf M}}\sup_{z\in k{\mathbf A}} \frac{ \|F(x+z)-F(x)-Tz\|_{{\mathbf H}_2}}{\|z\|_{{\mathbf H}_1}}, \ k\ge 1.
\end{equation}
 In Section \ref{sparserieszproperty.section}, we show that
 $F$ has the sparse Riesz property \eqref{sap.def} when
 $$\gamma_{F, T}(2{\bf A})< \frac{\sqrt{2}}{2}-\sqrt{\delta_{2{\bf A}}(T)} <\sqrt{1-\delta_{2\mathbf A}(T)},   $$
 and the almost linear property \eqref{almostlinearproperty}
 when
 $\gamma_{F, T}(4{\bf A})< \infty$.  Therefore
  signals $x\in {\bf M}$
could be reconstructed  from their nonlinear measurements $F(x)$ when $F$ is not far from a linear operator $T$
with small restricted isometry  constant $\delta_{2{\bf A}}(T)$, see Theorem \ref{optimization.thm2}.

\section{Iterative Reconstruction Algorithms}\label{algorithm.section}

For a Banach/Hilbert  space ${\mathbf B}$, we also denote its norm  by $\|\cdot\|$ for brevity.
In this section, we  establish
exponential convergence of the iterative reconstruction algorithms
\eqref{banachalgorithm1.intro}  and
 \eqref{vancittert.algorithm.intro}.

\begin{thm}\label{banachalgorithm.thm}
Let  ${\mathbf B}_1$ and ${\mathbf B}_2$ be Banach spaces,
$F$ be a differentiable map from ${\mathbf B}_1$ to ${\mathbf B}_2$
with its derivative being continuous and uniformly stable,
and let
$T\in {\mathcal B}({\mathbf B}_1, {\mathbf B}_2)$
be bounded below. 
 Assume that \eqref{banachalgorithm.thm.eq1}
 holds for some  bounded left-inverse
$T^\dag: {\mathbf B}_2\to {\mathbf B}_1$  of the linear operator $T$.
 Given positive relaxation factor $\mu>0$, an initial $x_0\in {\mathbf B}_1$ and a noisy observation data $z_\epsilon:=F(x^0)+\epsilon\in {\mathbf B}_2$  for some
 $x^0\in {\mathbf B}_1$ with additive noise $\epsilon\in {\mathbf B}_2$, define $x_n, n\ge 1$, iteratively by
 \eqref{banachalgorithm1.intro}.
Then  $x_n, n\ge 0$,
converges
exponentially to some $x_\infty\in {\mathbf B}_1$ with
\begin{equation}  \label{banachalgorithm.thm.eq2}
\|x_\infty-x^0\|\le \frac{\|T^\dag\|}{1- \beta_{F,T}\|T\| \|T^\dag\| }
\Big(\inf_{x\in {\mathbf B}_1}\inf_{0\ne y\in {\mathbf B}_1}
\frac{\|F'(x)y\|}{\|Ty\|}\Big)^{-1}  \|\epsilon \|,
\end{equation}
provided that
\begin{equation}\label{alphacondition}
0<\mu\le \Big(\sup_{ x\in {\mathbf B}_1}\sup_{0\ne y\in {\mathbf B}_1}
 \frac {\|F'(x)y\|}{\|Ty\|}\Big)^{-1}.\end{equation}
 Moreover,
\begin{equation}\label{banachalgorithm.thm.eq4}
 \|x_n- x_\infty\| \le   \frac{\|T^\dag\| \|F(x_0)-z_\epsilon\|}{1-\beta_{F,T}\|T\|\|T^\dag\|}
 \Big(\inf_{x\in {\mathbf B}_1}\inf_{0\ne y\in {\mathbf B}_1}
\frac{\|F'(x)y\|}{\|Ty\|}\Big)^{-1}
r_0^n,\quad n\ge 1,
\end{equation}
where
$$r_0=1-\mu (1- \beta_{F,T}\|T\|\|T^\dag\| )
\Big(\inf_{x\in {\mathbf B}_1}\inf_{0\ne y\in {\mathbf B}_1}
\frac{\|F'(x)y\|}{\|Ty\|}\Big)
\in (0, 1).$$
 \end{thm}

\begin{proof}
Set
$\alpha_n=\int_0^1 \|F'(x_{n-1}+t(x_n-x_{n-1}))(x_n-x_{n-1})\|dt, n\ge 1$. Then for $n\ge 1$,
\begin{eqnarray*}\label{banachalgorithm.thm.pf.eq1}
& & \|x_{n+1}-x_n\| \nonumber\\
& = & \|(x_n-x_{n-1})-\mu  T^\dag (F(x_n)-F(x_{n-1}))\|\nonumber\\
&\le & \Big\| x_n-x_{n-1}-\mu\alpha_n
\frac{T^\dag T(x_n-x_{n-1})}{\|T(x_n-x_{n-1})\|}\Big\| \nonumber\\
& &   +
\mu  \int_0^1 \|F'(x_{n-1}+t(x_n-x_{n-1}))(x_n-x_{n-1})\|\nonumber\\
& & \quad \times \Big\| T^\dag \Big(
\frac {F'(x_{n-1}+t(x_n-x_{n-1}))(x_n-x_{n-1})}{ \|F'(x_{n-1}+t(x_n-x_{n-1}))(x_n-x_{n-1})\|}-\frac{T(x_n-x_{n-1})}{\|T(x_n-x_{n-1})\|}\Big)\Big\| dt
\nonumber\\
&\le & \Big(1-\mu \frac{\alpha_n}{\|T(x_n-x_{n-1})\|}\Big)\|x_n-x_{n-1}\|   +
\mu \alpha_n \beta_{F, T}\|T^\dag\|
\nonumber\\
& \le &
\Big(1-\mu (1-\beta_{F,T}\|T\|\|T^\dag\| )
 \frac{\alpha_n}{\|T(x_n-x_{n-1})\|}\Big)\|x_n-x_{n-1}\|\nonumber\\
& \le & r_0 \|x_n-x_{n-1}\|
\end{eqnarray*}
by  \eqref{banachalgorithm1.intro}, \eqref{banachalgorithm.thm.eq1}  and \eqref{alphacondition}.
This proves  the exponential convergence of $x_n, n\ge 0$, to its limit  $ x_\infty\in {\mathbf B}_1$.

Taking limit in \eqref{banachalgorithm1.intro} gives 
   \begin{equation} \label{banachalgorithm.thm.pf.eq2}
 T^\dag (F( x_\infty)-F(x^0))=T^\dag \epsilon,
 \end{equation}
because  \begin{equation*}
 \|T^\dag(F(x_\infty)-z_\epsilon)\|\le  \|T^\dag(F(x_\infty)-F(x_n))\|+ \|x_{n+1}-x_n\|/\mu\to 0, \ n\to \infty.
  \end{equation*}
Then it follows from \eqref{banachalgorithm.thm.eq1}, \eqref{alphacondition} and \eqref{banachalgorithm.thm.pf.eq2} that
 \begin{eqnarray}
& &
\Big(\frac{\int_0^1 \|F'(x_\infty+t(x^0-x_\infty)) (x^0-x_\infty) \| dt}
{\|T(x^0-x_\infty)\|}\Big) 
 \|x^0-x_\infty\|\nonumber\\
& = &
\Big\| T^\dag \Big(\int_0^1 \|F'(x_\infty+t(x^0-x_\infty)) (x^0-x_\infty) \|
\frac{  T(x^0-x_\infty)}
{\|T(x^0-x_\infty)\|} dt \Big)\Big\|
\nonumber\\
& \le &\beta_{F,T} \|T^\dag\| \Big(\int_0^1 \|F'(x_\infty+t(x^0-x_\infty)) (x^0-x_\infty) \| dt\Big)\nonumber\\
& & + \Big\|T^\dag\int_0^1 F'(x_\infty+t(x^0-x_\infty)) (x^0-x_\infty)  dt\Big\|\nonumber\\
&\le  & \beta_{F,T} \|T\|\|T^\dag\|
 \Big(\frac{\int_0^1 \|F'(x_\infty+t(x^0-x_\infty) (x^0-x_\infty) \| dt}{\|T(x^0-x_\infty)\|}\Big)\nonumber\\
 & & \times
 \|x^0-x_\infty\|+\|T^\dag\| \|\epsilon\|,
 \end{eqnarray}
 which proves \eqref{banachalgorithm.thm.eq2}.

Observe that
 \begin{equation*}
 \|x_n-x_\infty\|\le \sum_{k=n}^\infty \|x_{k+1}-x_k\|\le
 \frac{\|x_1-x_0\|}{1-r_0} r_0^n\le \frac{\mu \|T^\dag\|\|F(x_0)-z_\epsilon\|}{1-r_0} r_0^n.
 \end{equation*}
Then the estimate  \eqref{banachalgorithm.thm.eq4} follows.
  \end{proof}

The iterative algorithm \eqref{banachalgorithm1.intro} in
Theorem \ref{banachalgorithm.thm}  provides a stable  reconstruction
of  $x\in {\mathbf B}_1$ from its noisy observation $F(x)+\epsilon\in {\mathbf B}_2$
when $\beta_{F,T}<(\|T\|\|T^\dag\|)^{-1}$, a requirement stronger than $\beta_{F,T}<1$
 that guarantees the bi-Lipschitz property for the map $F$, see  Theorem \ref{frame.sufficientcondition}
 in Appendix \ref{bilipschitz.appendix}.
Next we 
close that requirement gap on   $\beta_{F,T}$  in Hilbert space setting, cf. Theorem \ref{frame.sufficientcondition.hilbert}. 

\begin{thm}\label{hilbertalgorithm.thm}
Let  ${\mathbf H}_1$ and ${\mathbf H}_2$ be  Hilbert spaces,
$F$ be a differentiable map from ${\mathbf H}_1$ to ${\mathbf H}_2$
with its derivative being continuous and  satisfying
\eqref{localstability.thm.eq1}, and let
$T\in {\mathcal B}({\mathbf H}_1, {\mathbf H}_2)$ satisfy \eqref{Tstable.condition}
and \eqref{frame.sufficientcondition.hilbert.eq1}.
 Given relaxation  factor $\mu >0$, an initial $u_0\in {\mathbf H}_1$, and noisy data $z_\epsilon:=F(u^0)+\epsilon$ for some $u^0\in {\mathbf H}_1$
  with additive noise $\epsilon\in {\mathbf H}_2$, define $u_n, n\ge 1$, iteratively by
  \eqref{vancittert.algorithm.intro}.
Then  $u_n, n\ge 0$,
converges
exponentially to some $u_\infty\in {\mathbf H}_1$ with
\begin{equation}  \label{hilbertalgorithm.thm.eq2}
  \|u_\infty-u^0\|\le \frac{2\|T\|}{(2-\beta_{F,T}^2)\big(\inf_{\|u\|=1} \|Tu\|\big)\big(\inf_{v\in {\mathbf H}_1}\inf_{\|u\|=1} \|F'(v)u\|\big)}
 \|\epsilon\|,
\end{equation}
provided that
\begin{equation}\label{betahilbertcondition}
0<\mu< (2-\beta_{F,T}^2) \frac{\big(\inf_{\|u\|=1} \|Tu\|\big)\big(\inf_{v\in {\mathbf H}_1}\inf_{\|u\|=1} \|F'(v)u\|\big)}
{\|T\|^2\big(\sup_{v\in {\mathbf H}_1}\|F'(v)\|\big)^2}.\end{equation}
\end{thm}

\begin{proof} Define $S:=T^*F$.
Observe that
\begin{eqnarray*}
\langle F'(u)v, Tv\rangle & = &  \|F'(u) v\|\|Tv\|
\Big( 1-\frac{1}{2} \Big\|\frac{F'(u)v}{\|F'(u)v\|}-\frac{Tv}{\|Tv\|}\Big\|^2\Big)\\
&\ge & \frac{2-(\beta_{F,T})^2}{2} \|F'(u) v\|\|Tv\|.  
\end{eqnarray*}
Therefore
\begin{eqnarray} \label{hilbertalgorithm.thm.pf.eq1}
 & &  \langle v_1-v_2, S(v_1)-S(v_2)\rangle
  =  \langle F(v_1)-F(v_2), T(v_1-v_2)\rangle\nonumber\\
  & = & \int_0^1 \langle F'(v_2+t(v_1-v_2)(v_1-v_2), T(v_1-v_2)\rangle  dt\nonumber \\
& \ge &  \frac{2-(\beta_{F,T})^2}{2} \Big(\int_0^1 \|F'(v_2+t(v_1-v_2) (v_1-v_2)\| dt\Big)\|T(v_1-v_2)\| \nonumber\\
& \ge &  
  \frac{2-(\beta_{F,T})^2}{2}
 \Big(\inf_{\|u\|=1} \|Tu\|\Big)\Big(\inf_{v\in {\mathbf H}_1}\inf_{\|u\|=1} \|F'(v)u\|\Big)
  \|v_1-v_2\|^2, \end{eqnarray}
  where $A$ is the lower stability bound in \eqref{localstability.thm.eq1}.
Also one may easily verify that
 \begin{equation} \label{hilbertalgorithm.thm.pf.eq2}
\|S(v_1)-S(v_2)\|\le  \|T\|\big(\sup_{v\in {\mathbf H}_1} \|F'(v)\|\big) \|v_1-v_2\|, \ v_1, v_2\in {\mathbf H}_1. \end{equation}
Therefore by standard arguments (see for instance  \cite{nonlinearbook}), we obtain from  \eqref{hilbertalgorithm.thm.pf.eq1} and \eqref{hilbertalgorithm.thm.pf.eq2} that
\begin{equation*}
\|u_{n+1}-u_n\|^2  
 \le   r_1\|u_n-u_{n-1}\|^2, \ n\ge 1,
 \end{equation*}
 where
 \begin{eqnarray*} r_1 & = & 1-\mu (2-\beta_{F,T}^2)
 \Big(\inf_{\|u\|=1} \|Tu\|\Big)\Big(\inf_{v\in {\mathbf H}_1}\inf_{\|u\|=1} \|F'(v)u\|\Big)\nonumber\\
 & & \quad
 +\mu^2  \|T\|^2\Big(\sup_{v\in {\mathbf H}_1}\|F'(v)\|\Big)^2\in (0, 1).
   \end{eqnarray*}
This proves the exponential convergence of $u_n, n\ge 0$, in  the iterative algorithm \eqref{vancittert.algorithm.intro}.

 Taking limit in the  algorithm \eqref{vancittert.algorithm.intro}  leads to
 \begin{equation}\label{hilbertalgorithm.thm.pf.eq3}
T^* (F(u_\infty)-w_\epsilon)=0, \end{equation}
 where $u_\infty$ is  the limit of  the sequence $u_n, n\ge 0$.
 Thus
 \begin{eqnarray*}
& &   \frac{2-\beta_{F,T}^2}{2}
 \big(\inf_{\|u\|=1} \|Tu\|\big)\big(\inf_{v\in {\mathbf H}_1}\inf_{\|u\|=1} \|F'(v)u\|\big)
\|u_\infty-u^0\|^2\\
& \le &
\langle  u_\infty- u^0,  T^* (F(u_\infty)-F(u^0))\rangle
=
\langle u^0-u_\infty,  T^* \epsilon\rangle\nonumber\\
& \le &  \|T\| \|u^0-u_\infty\|\|\epsilon\|
 \end{eqnarray*}
 by \eqref{hilbertalgorithm.thm.pf.eq1} and \eqref{hilbertalgorithm.thm.pf.eq3}. This proves
 \eqref{hilbertalgorithm.thm.eq2}
 and completes the proof.
\end{proof}

\section{Iterative algorithm for localized maps}
\label{localizedmaps.section}

In this section, we develop a fixed point theorem   for a
well-localized map on a Banach space
whose restriction on its dense Hilbert subspace is a contraction,
and we establish   exponential convergence of
the iterative algorithm \eqref{vancittert.algorithm.intro}
 for  certain localized maps between  (weak-)Hilbert-dense
 Banach spaces.

To state our results, we recall the concept of differential subalgebras.
Given two unital Banach algebras ${\mathcal A}_1$ and ${\mathcal A}_2$,
 ${\mathcal A}_1$ is  said to be a {\em Banach  subalgebra}  of  ${\mathcal A}_2$  if
${\mathcal A}_1\subset {\mathcal A}_2$,
 ${\mathcal A}_1$ and ${\mathcal A}_2$ share the same identity
and
$\sup_{0\ne T\in {\mathcal A}_1} \|T\|_{{\mathcal A}_2}/\|T \|_{{\mathcal A}_1}<\infty$ holds;
and  a Banach subalgebra ${\mathcal A}_1$  of ${\mathcal A}_2$  is said to be a {\em differential subalgebra} of order $\theta\in (0, 1]$ if
there exists  a positive constant $D$
such that
\begin{equation}\label{dnp.def}
\|T_1T_2\|_{{\mathcal A}_1} \le D \|T_1\|_{{\mathcal A}_1} \|T_2\|_{{\mathcal A}_1}
 \Big(\Big(\frac{\|T_1\|_{{\mathcal A}_2}}{\|T_1\|_{{\mathcal A}_1}}\Big)^{\theta}
+\Big (\frac{\|T_2\|_{{\mathcal A}_2}}{\|T_2\|_{{\mathcal A}_1}}\Big)^{\theta}\Big)\end{equation}
for all nonzero $T_1, T_2\in {\mathcal A}_1$ \cite{blackadarcuntz91,   kissin94,   rieffel10, sunaicm13}.
We remark that differential subalgebras include  many families of Banach algebras of infinite matrices with certain off-diagonal decay and
localized integral operators
\cite{
gltams06,  jaffard90, shincjfa09,  suncasp05,
suntams07,   sunacha08, sunca11, sunaicm13},
and they
 have been
  widely used in operator theory, non-commutative geometry,
  frame theory, algebra of pseudodifferential operators,  numerical analysis, signal processing, control and  optimization etc,
see \cite{blackadarcuntz91, christensen05, halljin10, kissin94,   mjieee08, Nashed10, 
rieffel10,  sunsiam06, sunaicm13},
  the  survey papers \cite{grochenigsurvey, Krishtal11} 
  and references therein.

Next we define  the conjugate $T^*$
 of a localized linear operator $T$
 between  (weak-) Hilbert-dense
 Banach spaces.
Given Banach spaces ${\mathbf B}_i$ and their dense Hilbert subspaces ${\mathbf H}_i, i=1, 2$,
 we assume that linear operators $T$ reside in a Banach subspace ${\mathcal B}$ of ${\mathcal B}({\mathbf H}_1, {\mathbf H}_2)$ and
also of ${\mathcal B}({\mathbf B}_1, {\mathbf B}_2)$. For a linear operator $T\in {\mathcal B}$, its restriction
$T|_{{\mathbf H}_1}$ to
${\mathbf H}_1$ is a  bounded operator from ${\mathbf H}_1$ to ${\mathbf H}_2$, hence its conjugate 
$(T|_{{\mathbf H}_1})^*$ is well-defined on ${\mathbf H}_2$, and the ``conjugate" of $T$ is well-defined  if the conjugate
$(T|_{{\mathbf H}_1})^*$  can be extended to a bounded operator from ${\mathbf B}_2$ to ${\mathbf B}_1$.
The above approach to define the conjugate requires certain localization  for linear operators in  ${\mathcal B}$, which
will be stated precisely
in the  next theorem, cf. \eqref{adjoint.banach.def}.

\begin{thm}\label{banachwiener.algorithm.thm}
Let  ${\mathbf B}_1$ and ${\mathbf B}_2$ be Banach spaces,
 ${\mathbf H}_1$ and ${\mathbf H}_2$ be Hilbert spaces with the property that for $i=1, 2$,
 ${\mathbf H}_i\subset {\mathbf B}_i$, ${\mathbf H}_i$ is dense in  ${\mathbf B}_i$, and
\begin{equation}
\sup_{0\ne x\in {\mathbf H}_i} \frac{\|x\|_{{\mathbf B}_i}}{\|x\|_{{\mathbf H}_i}}<\infty.
\end{equation}
Assume that Banach algebra
${\mathcal A}$ with norm $\|\cdot\|_{\mathcal A}$ is a unital Banach subalgebra of
 ${\mathcal B}({\mathbf B}_1)$  and a differential subalgebra of ${\mathcal B}({\mathbf H}_1)$ of order $\theta\in (0,1]$.
Let  ${\mathcal B}$ be a Banach subspace of both
${\mathcal B}({\mathbf H}_1, {\mathbf H}_2)$ and  ${\mathcal B}({\mathbf B}_1, {\mathbf B}_2)$,
and
${\mathcal B}^*$ be a Banach subspace of both
${\mathcal B}({\mathbf H}_2, {\mathbf H}_1)$
and  ${\mathcal B}({\mathbf B}_2, {\mathbf B}_1)$
 such that
\begin{itemize}
\item [{(i)}] $ST\in {\mathcal A}$ for all $S\in {\mathcal B}^*$ and $T\in {\mathcal B}$. Moreover,
\begin{equation}\label{compose.banachalgebra}
\sup_{0\ne T\in {\mathcal B}, 0\ne S\in {\mathcal B}^*} \frac{\|ST\|_{\mathcal A}}{\|S\|_{{\mathcal B}^*}\|T\|_{\mathcal B}}<\infty.\end{equation}
\item[{(ii)}] For any $T\in {\mathcal B}$ and $S\in {\mathcal B}^*$ there exist unique $T^*\in {\mathcal B}^*$ and $S^*\in {\mathcal B}$ with the property that
$\|T^*\|_{{\mathcal B}^*}=\|T\|_{\mathcal B}, \|S^*\|_{{\mathcal B}}=\|S\|_{{\mathcal B}^*}$ and
\begin{equation}\label{adjoint.banach.def}
\langle Tu, w\rangle=\langle u, T^*w\rangle \ {\rm and} \ \langle Sw, u\rangle=\langle w, S^*u\rangle \quad {\rm for \ all} \ u\in {\mathbf H}_1\ {\rm and} \ w\in {\mathbf H}_2.
\end{equation}
\end{itemize}
Assume that $F$ is a differentiable map from ${\mathbf B}_1$ to ${\mathbf B}_2$ such that
its derivative $F'$ is continuous and bounded from  ${\mathbf B}_1$ into
 ${\mathcal B}$,  and
\begin{equation}
\beta_{F, T}=\sup_{x\in {\mathbf B}_1} \sup_{u\in {\mathbf H}_1}
\Big\|\frac{F'(x) u}{\|F'(x)u\|_{{\mathbf H}_2}} -\frac{Tu}{\|Tu\|_{{\mathbf H}_2}}\Big\|_{{\mathbf H}_2}<\sqrt{2}
\end{equation}
for some  linear operator $T\in {\mathcal B}$.
Take an initial $x_0\in {\mathbf B}_1$ and a noisy observation data $z_\epsilon=F(x^0)+\epsilon$ for some $x^0\in {\mathbf B}_1$ and $\epsilon\in {\mathbf B}_2$,
define
$x_n, n\ge 1$, iteratively by
\begin{equation}
\label{vancittert.banachalgorithm}
x_{n+1}=x_n-\mu T^*(F(x_n)-z_\epsilon), \ n\ge 0,
\end{equation}
where the relaxation  factor $\mu$ satisfies 
\begin{equation}\label{betahilbertbanachcondition}
0<\mu< (2-\beta_{F,T}^2) \frac{\big(\inf_{\|u\|_{{\mathbf H}_1}=1} \|Tu\|_{{\mathbf H}_2}\big)\big(\inf_{x\in {\mathbf B}_1}\inf_{\|u\|_{{\mathbf H}_1}=1} \|F'(v)u\|_{{\mathbf H}_2}\big)}
{\big(\sup_{\|u\|_{{\mathbf H}_1}=1}\|Tu\|_{{\mathbf H}_2}\big)^2\big(\sup_{x\in {\mathbf B}_1}\sup_{\|u\|_{{\mathbf H}_1}=1}\|F'(v)u\|_{{\mathbf H}_2}\big)^2},\end{equation}
 and $T^*\in {\mathcal B}^*$ is the conjugate operator defined by \eqref{adjoint.banach.def}.
Then $x_n, n\ge 0$, converges  exponentially  to some $x_\infty\in {\mathbf B}_1$
with
\begin{equation}\label{banachwiener.algorithm.thm.errorestimate}
\|x_\infty-x^0\|_{{\mathbf B}_1}\le C \|\epsilon\|_{{\mathbf B}_2},
\end{equation}
where $C$ is an absolute positive constant.
 \end{thm}

Given a Banach space ${\mathbf B}$, we say that a map $G: {\mathbf B}\to {\mathbf B}$ is a {\em contraction}
if there exists $r\in [0, 1)$ such that
\begin{equation*}
\|G(x)-G(y)\|_{{\mathbf B}}\le  r \|x-y\|_{{\mathbf B}}\quad {\rm  for \ all} \ x, y\in {\mathbf B}.\end{equation*}
For a contraction $G$ on a Banach space ${\mathbf B}$, the Banach fixed point theorem states
that there is a unique fixed point $x^*$ for the contraction $G$ (i.e., $G(x^*)=x^*$), and for any initial $x_0\in {\mathbf B}$,
the sequence $x_{n+1}=G(x_n), n\ge 0$, converges  exponentially to the fixed point  $x^*$  \cite{dieu69}.
To prove Theorem \ref{banachwiener.algorithm.thm}, we need a  fixed point theorem for differentiable maps  on a Banach space
with its derivative being continuous and bounded in a differential Banach subalgebra and its restriction on a dense Hilbert subspace being a contraction.

\begin{thm}\label{fixedpoint.thm}
Let ${\mathbf B}$ be a Banach space, ${\mathbf H}$ be a Hilbert space such that ${\mathbf H}\subset {\mathbf B}$ is dense in ${\mathbf B}$
and
\begin{equation}\label {fixedpoint.lem.eq1}\sup_{0\ne x\in {\mathbf H}}\frac{\|x\|_{\mathbf B}}{\|x\|_{\mathbf H}}<\infty,
\end{equation}
and let ${\mathcal A}$ be a Banach subalgebra of  ${\mathcal B}({\mathbf B})$
and also a differential subalgebra of ${\mathcal B}({\mathbf H})$ of order $\theta\in (0, 1]$.
If $G$ is a differentiable map on ${\mathbf B}$ whose derivative  $G'$ is continuous and bounded from ${\mathbf B}$ into ${\mathcal A}$ 
and there exists $r\in [0, 1)$ such that
\begin{equation}\label{fixedpoint.lem.eq2}
 \|G'(x)\|_{{\mathcal B}({\mathbf H})}\le r\ {\rm for \ all} \ x\in {\mathbf B},
\end{equation}
then
there exists a unique fixed point $x^*$ for the map $G$. Furthermore given any initial $x_0\in {\mathbf B}$, the sequence
$x_n, n\ge 0$, defined by
\begin{equation} \label{fixedpoint.lem.eq3}
x_{n+1}=G(x_n), \ n\ge 0,
\end{equation}
converges exponentially to the fixed point $x^*$.
\end{thm}

\begin{proof}
Let $x_n, n\ge 0$, be as in \eqref{fixedpoint.lem.eq3}. It follows from  the continuity of $G'$
in the Banach subalgebra ${\mathcal A}$ of ${\mathcal B}({\mathbf B})$ that
\begin{eqnarray}\label{fixedpoint.lem.pf.eq1}
x_{n+1}-x_n&= &G(x_n)-G(x_{n-1})\nonumber\\
&= & \Big(\int_0^1 G'(x_{n-1}+t(x_n-x_{n-1}))dt\Big) (x_n-x_{n-1})\nonumber\\
&=: & T_n(x_n-x_{n-1}), \ n\ge 1.
\end{eqnarray}
Observe that
\begin{equation} \label{fixedpoint.lem.pf.eq2}
\|T_n\|_{{\mathcal B}({\mathbf H})}\le
\int_0^1 \| G'(x_{n-1}+t(x_n-x_{n-1}))\|_{{\mathcal B}({\mathbf H})} dt\le r
\end{equation}
and
\begin{equation} \label{fixedpoint.lem.pf.eq3}
\|T_n\|_{\mathcal A}\le \int_0^1 \| G'(x_{n-1}+t(x_n-x_{n-1}))\|_{\mathcal A} dt\le M
\end{equation}
where $M=\sup_{x\in {\mathbf B}} \|G'(x)\|_{\mathcal A}<\infty$ by the assumption on the map $G$.
Set
$$b_n=\sup_{l\ge 1} \|T_{l+n-1} T_{l+n-2} \cdots T_l\|_{\mathcal A},\ n\ge 1.$$
Then we obtain from   \eqref{dnp.def}, \eqref{fixedpoint.lem.pf.eq2} and
\eqref{fixedpoint.lem.pf.eq3}
that
\begin{equation*} 
b_{2n+1}\le  \big(\sup_{m\ge 1} \|T_m\|_{\mathcal A}\big)  b_{2n}\le M b_{2n}\end{equation*}
and
 \begin{equation*} 
  b_{2n} \le  2D
\big(\sup_{l\ge 1} \|T_{l+n-1} \cdots T_l\|_{{\mathcal B}({\mathbf H})}
\big)^{\theta}
 (b_n)^{2-\theta}\le  2D r^{n\theta }
 (b_n)^{2-\theta}
\end{equation*}
for all $n\ge 1$. Thus
\begin{eqnarray*}
b_n & \le  &  M^{\epsilon_0} b_{n-\epsilon_0}\le
M^{\epsilon_0} (2D)  r^{\theta(n-\epsilon_0)/2}
(b_{(n-\epsilon_0)/2})^{2-\theta}\\
& \le &
M^{\epsilon_0+(2-\theta) \epsilon_1} (2D)^{1+(2-\theta)}
 r^{\frac{\theta}{2} ((n-\epsilon_0)+(n-\epsilon_0-2\epsilon_1)\frac{2-\theta}{2})}
(b_{(n-\epsilon_0-2\epsilon_1)/4})^{2-\theta}\nonumber\\
&\le &
  \cdots\\
  &\le &   M^{\sum_{i=0}^{l} \epsilon_i (2-\theta)^i}
  (2D)^{\sum_{i=0}^{l-1} (2-\theta)^i}
r^{\frac{\theta}{2} \sum_{i=0}^{l-1}\sum_{j=i+1}^l \epsilon_j 2^{j-i} (2-\theta)^i},
\end{eqnarray*}
where $n=\sum_{i=0}^l \epsilon_i 2^i$  with $\epsilon_i\in \{0, 1\}$ and $\epsilon_l=1$.
Therefore
\begin{equation}\label{fixedpoint.lem.pf.eq4}
b_n\le \big((2D)^{1/(1-\theta)} (M/r_0)^{(2-\theta)/(1-\theta)}\big)^{n^{\log_2(2-\theta)}}
r^n, \ n\ge 1
\end{equation}
if $\theta\in (0,1)$, and
\begin{equation}\label{fixedpoint.lem.pf.eq5}
b_n\le \frac{M}{r}
 \big(2D M/r\big)^{\log_2 n}
r^n, \ n\ge 1
\end{equation}
if $\theta=1$.  By \eqref{fixedpoint.lem.pf.eq4} and \eqref{fixedpoint.lem.pf.eq5},
for any $r_1\in (r, 1)$ there exists a positive constant $C$
such that
\begin{equation} \label{fixedpoint.lem.pf.eq6}
\|T_nT_{n_1}\cdots T_1\|_{\mathcal A}\le C r_1^n,\  n\ge 1.
\end{equation}
Recall that ${\mathcal A}$ is a Banach subalgebra of ${\mathcal B}({\mathbf B})$.
We then obtain from \eqref{fixedpoint.lem.pf.eq1} and \eqref{fixedpoint.lem.pf.eq6}
that
\begin{equation*}
\|x_{n+1}-x_n\|_{\mathbf B} \le C r_1^n \|x_1-x_0\|_{\mathbf B},\ n\ge 1,
\end{equation*}
which proves the exponential convergence of the sequence $x_n, n\ge 0$.

Finally we prove the uniqueness of the fixed point for the map $G$. Let $x^*$ and $\tilde x^*$ be fixed points of the map $G$. Then
$x^*-\tilde x^*$ is a fixed point of the linear operator $T:=\int_0^1 G'(\tilde x^*+t(x^*-\tilde x^*)) dt\in {\mathcal A}$, because
\begin{equation} \label{fixedpoint.lem.pf.eq7}
x^*-\tilde x^*= G(x^*)-G(\tilde x^*)= T (x^*-\tilde x^*).
\end{equation}
Following the argument to prove
\eqref{fixedpoint.lem.pf.eq6}, we obtain that
$\lim_{n\to \infty} \|T^n\|_{{\mathcal B}({\mathbf B})}=0$. This together with \eqref{fixedpoint.lem.pf.eq7}  implies that
$x^*=\tilde x^*$, the uniqueness of  fixed points for the map $G$.
\end{proof}

Now we apply Theorem \ref{fixedpoint.thm} to prove Theorem \ref{banachwiener.algorithm.thm}.

\begin{proof} [Proof of Theorem \ref{banachwiener.algorithm.thm}]
Define $G:{\mathbf B}_1\to {\mathbf B}_1$ by
\begin{equation}
G(x)= x-\mu T^*(F(x)-z_\epsilon),\  x\in {\mathbf B}_1.
\end{equation}
Then
$G$ is differentiable on ${\mathbf B}_1$ and its derivative
$G'(x)=I-\mu T^* F'(x), x\in {\mathbf B}_1$,
is continuous and bounded in ${\mathcal A}$ by the assumption on $F$ and the Banach spaces ${\mathcal B}$ and ${\mathcal B}^*$.
Set
$$m_0= \frac{2-\beta_{F,T}^2}{2}
  \Big(\inf_{0\ne u\in {\mathbf H}_1}\frac{\|Tu\|_{{\mathbf H}_2}}{\|u\|_{{\mathbf H}_1}}\Big) \Big(\inf_{x\in {\mathbf B}_1}\inf_{0\ne u\in {\mathbf H}_1}
  \frac{\|F'(x)u\|_{{\mathbf H}_2}}{\|u\|_{{\mathbf H}_1}}\Big)$$
  and
 $$ M_0=
  \Big(\sup_{0\ne u\in {\mathbf H}_1}\frac{\|Tu\|_{{\mathbf H}_2}}{\|u\|_{{\mathbf H}_1}}\Big) \Big(\sup_{x\in {\mathbf B}_1}\sup_{0\ne u\in {\mathbf H}_1}
  \frac{\|F'(x)u\|_{{\mathbf H}_2}}{\|u\|_{{\mathbf H}_1}}\Big).
 $$
Observe that
\begin{eqnarray} \label{banachwiener.algorithm.thm.pf.eq1}
 \|G'(x)\|_{\mathcal A} & \le &
\|I\|_{\mathcal A}+\mu  \|T\|_{{\mathcal B}}\big(\sup_{x\in {\mathbf B}_1} \|F'(x)\|_{\mathcal B}\big)\nonumber\\
& & \quad \times
\Big (\sup_{0\ne S_1,S_2\in {\mathcal B}} \frac{\|S_1^*S_2\|_{{\mathcal A}}}{\|S_1\|_{\mathcal B}\|S_2\|_{\mathcal B}}\Big)<\infty,
\end{eqnarray}
and
\begin{eqnarray} \label{banachwiener.algorithm.thm.pf.eq2}
\|G'(x)\|_{{\mathcal B}({\mathbf H}_1)} & \le &
\|(I+\mu  T^* F'(x))^{-1}\|_{{\mathcal B}({\mathbf H}_1)}
\|1-\mu^2 (T^*F'(x))^2 \|_{{\mathcal B}({\mathbf H}_1)}\nonumber\\
&\le & (1+M_0^2\mu^2)
\sup_{0\ne u\in {\mathbf H}_1} \frac{\|u\|_{ {\mathbf H}_1}^2}{\langle u, (1+\mu T^* F'(x)) u\rangle_{ {\mathbf H}_1}}\nonumber\\
 &\le &   \frac{1+M_0^2\mu^2}{1+ m_0\mu}<1, \ x\in {\mathbf B}_1,
\end{eqnarray}
where the second inequality holds as
$$\|(I+\mu T^* F'(x))^{-1}\|_{{\mathcal B}({\mathbf H}_1)}=\sup_{0\ne u\in {\mathbf H}_1} \frac{\|u\|_{{\mathbf H}_1}}{\|(I+\mu T^* F'(x))u\|_{{\mathbf H}_1}}$$
 and
the third inequality follows from
 \eqref{hilbertalgorithm.thm.pf.eq1}. 
Combining the above two estimates about $G'(x), x\in {\mathbf B}_1$, with Theorem \ref{fixedpoint.thm}
proves the exponential convergence of $x_n, n\ge 0$,
in ${\mathbf B}_1$.

Denote by $x_\infty$ the limit of $x_n, n\ge 0$, in ${\mathbf B}_1$.
Then taking limit in the iterative algorithm \eqref{vancittert.banachalgorithm} yields
\begin{equation*} 
T^* F(x_\infty)-T^* F(x^0)=T^* \epsilon.
\end{equation*}
Thus
\begin{equation}\label {banachwiener.algorithm.thm.pf.eq10}
A_\infty (x_\infty-x^0)=T^*\epsilon,
\end{equation}
where $A_\infty=\int_0^1 T^* F'(x^0+t(x_\infty-x^0)) dt$. Following the argument
to prove Theorem \ref{fixedpoint.thm} and applying \eqref{banachwiener.algorithm.thm.pf.eq1}
and \eqref{banachwiener.algorithm.thm.pf.eq2},
there exists a positive constant $C_r$ for any
 $r\in ( (1+M_0^2\mu^2/(1+m_0\mu), 1)$ such that
 \begin{equation*}
 \|(I-\mu A_\infty)^n\|_{\mathcal A}\le C_r r^n, \ n\ge 1.
 \end{equation*}
 Thus $A_\infty$ is invertible in ${\mathcal A}$ and
 \begin{equation}\label{banachwiener.algorithm.thm.pf.eq11}
 \|(A_\infty)^{-1}\|_{\mathcal A}\le \mu \sum_{n=0}^\infty \|(I-\mu A_\infty)^n\|_{\mathcal A}\le \mu
 (\|I\|_{\mathcal A}+ C_r/(1-r)).
 \end{equation}
 Combining \eqref{banachwiener.algorithm.thm.pf.eq10} and
 \eqref{banachwiener.algorithm.thm.pf.eq11} leads to
 \begin{eqnarray*}
  & & \|x_\infty-x^0\|_{{\mathbf B}_1} \le   \|(A_\infty)^{-1}\|_{{\mathcal B}({\mathbf B}_1)}\|T^* \epsilon\|_{{\mathbf B}_1}\\
 & \le & \|(A_\infty)^{-1}\|_{\mathcal A} \|T\|_{{\mathcal B}}
 \Big (\sup_{0\ne S\in {\mathcal A}} \frac{\|S\|_{{\mathcal B}({\mathbf B}_1)}
}{\|S\|_{\mathcal A}}\Big)
\Big (\sup_{0\ne U\in {\mathcal B}^*} \frac{\|U\|_{{\mathcal B}({\mathbf B}_2, {\mathbf B}_1))}
}{\|U\|_{\mathcal B^*}}\Big)\|\epsilon\|_{{\mathbf B}_2}.
 \end{eqnarray*} This proves the error estimate
 \eqref{banachwiener.algorithm.thm.errorestimate}.
\end{proof}

\begin{rem} {\rm Our model of Hilbert-dense
 Banach spaces in Theorem \ref{banachwiener.algorithm.thm} is $\ell^p(\Lambda)$, the space of $p$-summable sequences
  $\ell^p(\Lambda)$, with $2\le p<\infty$.
For that case, exponential convergence of the Van-Cittert algorithm,
which is similar to the iterative algorithm  \eqref{vancittert.banachalgorithm}
 in Theorem \ref{banachwiener.algorithm.thm}, is established in \cite{sunaicm13}
 under slightly different restriction on
the relaxation factor $\mu$.
For  weak-Hilbert-dense Banach spaces,
the  iterative algorithm  \eqref{vancittert.banachalgorithm}
 in Theorem \ref{banachwiener.algorithm.thm} still has exponential convergence
 if   operators in ${\mathcal B}$ and ${\mathcal B}^*$ are assumed additionally to be uniformly continuous
 in the weak topologies of Banach spaces, that is,
$\sup_{\|T\|_{\mathcal B}\le 1} |f(Tx_n)-f(Tx_\infty)|\to 0$
for any bounded linear functional $f$ on ${\mathbf B}_2$ if
 $x_n$ tends to $x_\infty$ in the weak topology of ${\mathbf B}_1$; and
$\sup_{\|S\|_{{\mathcal B}^*}\le 1} |g(Sy_n)-g(Sy_\infty)|\to 0$
for any bounded linear functional $g$ on ${\mathbf B}_1$ if
 $y_n$ tends to $y_\infty$ in the weak topology of ${\mathbf B}_2$.
 We  leave the detailed arguments to  interested readers. }
\end{rem}

\section{Sparse reconstruction and optimization}
\label{optimization.section}

In this section, we show that   sparse signals $x\in {\bf A}$ could be reconstructed from their nonlinear measurements $F(x)$ via
the optimization approach \eqref{optimizationsolution}.

\begin{thm}\label{optimization.thm1}
Let ${\mathbf H}_1$ and ${\mathbf H}_2$ be Hilbert spaces, ${\mathbf M}$ be a Banach space,
 ${\mathbf A}=\cup_{i\in I} {\mathbf A}_i$  be union of closed linear subspaces
 of ${\mathbf H}_1$, $s_{\mathbf A}$ and  $a_{\mathbf A}$  be in  \eqref{rationB1M} and
  \eqref{secondapproximationestimate} respectively, and
 let $F$ be a continuous map from ${\mathbf H}_1$ to ${\mathbf H}_2$ normalized so that $F(0)=0$.
 If $({\mathbf A}, {\mathbf M}, {\mathbf H}_1)$
forms a sparse approximation triple,
and
 if $F$ has the sparse Riesz property
\eqref{sap.def}  
and the almost linear  property \eqref{almostlinearproperty}
 with
  $D, \beta, \gamma_1, \gamma_2\ge 0$ satisfying
\begin{equation} \label{dbetagamma.condition}
\gamma_3:= 1-2D\gamma_1-(D\gamma_1+D\gamma_2+\beta)\sqrt{a_{\mathbf A}s_{\mathbf A}}>0,
\end{equation}
then given $x^0\in {\mathbf M}$ and $\varepsilon>0$,
the solution $x_{\mathbf M}^0$ of the optimization problem
\eqref{optimizationsolution}
provides a suboptimal  approximation to $x^0$,
\begin{equation}\label{optimization.thm1.error}
\|x_{\mathbf M}^0-x^0\|_{{\mathbf H}_1}  \le
\Big(   \frac{2+8D\gamma_2+4\beta}
{\gamma_3} \Big)
\sqrt{a_{\mathbf A}}
\sigma_{\mathbf A, \mathbf M}(x^0)\  +\
\frac{(2+\sqrt{a_{\mathbf A}s_{\mathbf A}})D}
{\gamma_3}
\varepsilon
\end{equation}
and
\begin{equation}\label{optimization.thm1.error2}
\|x_{\mathbf M}^0-x^0\|_{\mathbf M} \le
\Big( \frac{2-4D\gamma_1+2 (D\gamma_1+2D\gamma_2+\beta)\sqrt{a_{\mathbf A}s_{\mathbf A}}}
{\gamma_3}\Big)
\sigma_{\mathbf A, \mathbf M}(x^0) +\
\frac{2D}
{\gamma_3}
\sqrt{s_{\mathbf A}}
 \varepsilon.
\end{equation}
\end{thm}

To prove Theorem \ref{optimization.thm1}, we need the following approximation property for sparse approximation triples.

\begin{prop}\label{aA.prop} Let $({\bf A}, {\bf M}, {\bf H}_1)$ be a sparse approximation triple
and $a_{\bf A}$ be as in   \eqref{secondapproximationestimate}.
Then
\begin{equation}\label{b1Mtransfer}
\|x-x_{\mathbf A, \mathbf M}\|_{{\mathbf H}_1}
\le a_{\mathbf A}  \|x\|_{\mathbf M}, \ x\in {\mathbf M},
\end{equation}
where $x_{\mathbf A, \mathbf M}$ is a best approximator of $x\in {\bf M}$.
\end{prop}

We postpone the proof of the above proposition to Appendix \ref{sat.appendix} and start the proof of Theorem \ref{optimization.thm1}.

\begin{proof}[Proof of Theorem \ref{optimization.thm1}]
Let $x_{\mathbf A, \mathbf M}^0:={\rm argmin}_{\hat x\in {\mathbf A}} \|x^0-\hat x\|_{\mathbf M}$ be a best approximator  in ${\mathbf A}$ to $x^0$, where the existence follows from the proximinality property of the triple $({\bf A}, {\bf M}, {\bf H}_1)$.
Denote by ${\mathbf A}(x_{\mathbf A, \mathbf M}^0)$ the linear space in ${\mathbf A}$ containing $x_{\mathbf A, \mathbf M}^0$.
Then
\begin{equation}x_{\mathbf A, \mathbf M}^0={\rm argmin}_{\hat x\in {\mathbf A}(x_{\mathbf A, \mathbf M}^0)}\|x^0-\hat x\|_{\mathbf M}=
{\rm argmin}_{\hat x\in {\mathbf A}(x_{\mathbf A, \mathbf M}^0)}\|x^0-\hat x\|_{{\mathbf H}_1}
\end{equation}
by the common best approximator property \eqref{bestapproximation.hm};
and
\begin{equation}\label{optimization.thm1.pf.eq1}
 \|x^0\|_{\mathbf M}=\|x_{\mathbf A, \mathbf M}^0\|_{\mathbf M}+ \|x^0-x_{\mathbf A, \mathbf M}^0\|_{\mathbf M}
 =\|x_{\mathbf A, \mathbf M}^0\|_{\mathbf M}+\sigma_{\mathbf A, \mathbf M}(x^0)
\end{equation}
by  the norm splitting properties \eqref{optimization.thm1.eq3}  and \eqref{bestapproximation.hm.equiv}.

Let
 $x_{\mathbf A,\mathbf M}^0+h_0:={\rm argmin}_{\hat x\in {\mathbf A}(x_{\mathbf A, \mathbf M}^0)} \|x_{\mathbf M}^0-\hat x\|_{\mathbf M}
 \in {\mathbf A}(x_{\mathbf A, \mathbf M}^0)$
  be a best approximator to $x_{\mathbf M}^0$ in ${\mathbf A}(x_{\mathbf A, \mathbf M}^0)$.
Then
\begin{equation}
\label{optimization.thm1.pf.eq0}
x_{\mathbf A,\mathbf M}^0+h_0={\rm argmin}_{\hat x\in {\mathbf A}(x_{\mathbf A, \mathbf M}^0)} \|x_{\mathbf M}^0-\hat x\|_{{\mathbf H}_1}\end{equation}
and
\begin{equation} \label{optimization.thm1.pf.eq2}
 \|x_{\mathbf M}^0\|_{\mathbf M}=\|x_{\mathbf A, \mathbf M}^0+h_0\|_{\mathbf M}+ \|x^0_{\mathbf M}-x_{\mathbf A, \mathbf M}^0-h_0\|_{\mathbf M}
\end{equation}
by the common best approximator property and norm splitting property of the triple $({\bf A}, {\bf M}, {\bf H}_1)$. 

Set  $h:=x_{\mathbf M}^0-x^0$. We then obtain from \eqref{optimizationsolution}, 
\eqref{optimization.thm1.pf.eq1}  and \eqref{optimization.thm1.pf.eq2} that
\begin{eqnarray} \label{optimization.thm1.pf.eq4}
\|h-h_0\|_{\mathbf M} &\le & \|(x^0-x_{\mathbf A,\mathbf M}^0)+(h-h_0)\|_{\mathbf M}+\sigma_{\mathbf A, \mathbf M}(x^0)\nonumber\\
& = & \|x_{\mathbf M}^0\|_{\mathbf M}-\|x_{\mathbf A,\mathbf M}^0+h_0\|_{\mathbf M}+\sigma_{\mathbf A, \mathbf M}(x^0)
 \nonumber\\
 &\le &  \|x^0\|_{\mathbf M}-\|x_{\mathbf A,\mathbf M}^0+h_0\|_{\mathbf M}+\sigma_{\mathbf A, \mathbf M}(x^0)  
 \nonumber\\
&\le & \|h_0\|_{\mathbf M}+ 2\sigma_{\mathbf A, \mathbf M}(x^0).  
\end{eqnarray}

Let $h_1:={\rm argmin}_{\hat h\in {\mathbf A}} \|h-h_0-\hat h\|_{\mathbf M}$ be a best approximator of $h-h_0$.
Then
\begin{equation}
 \label{optimization.thm1.pf.eq2b-}
\|h-h_0\|_{{\mathbf H}_1}^2=\|h_1\|_{{\mathbf H}_1}^2+ \|h-h_0-h_1\|_{{\mathbf H}_1}^2,
\end{equation}
\begin{equation} \label{optimization.thm1.pf.eq2b}
 \|h-h_0\|_{\mathbf M}=\|h_1\|_{\mathbf M}+ \|h-h_0-h_1\|_{\mathbf M},
\end{equation}
and
\begin{equation} \label{optimization.thm1.pf.eq2c}
\|h-h_0-h_1\|_{{\mathbf H}_1}\le  \sqrt{a_{\mathbf A}} \|h-h_0\|_{\mathbf M}
\end{equation}
by  \eqref{bestapproximation.hm}, \eqref{optimization.thm1.eq3} and  \eqref{b1Mtransfer}.

From  
     \eqref{optimizationsolution}, \eqref{almostlinearproperty}, \eqref{b1Mtransfer} and \eqref{optimization.thm1.pf.eq4},
it follows that
\begin{eqnarray} \label{optimization.thm1.pf.eq5}
\|F(h)\| & \le & \|F(x_{\mathbf M}^0)-F(x^0)-F(h)\|+\varepsilon\nonumber\\
&\le &   \gamma_1 \|h\|_{{\mathbf H}_1}+ \gamma_2 \sqrt{a_{\mathbf A}}
(\sigma_{\mathbf A, \mathbf M}(x^0)+\sigma_{\mathbf A, \mathbf M}(x_{\mathbf M}^0))+\varepsilon\nonumber\\
& \le &
\gamma_1 \|h_0\|_{{\mathbf H}_1}+ \gamma_1 \|h_1\|_{{\mathbf H}_1}+ \gamma_1 \|h-h_0-h_1\|_{{\mathbf H}_1}\nonumber\\
& & +
\gamma_2 \sqrt{a_{\mathbf A}}
(\sigma_{\mathbf A, \mathbf M}(x^0)+ \|(x^0- x_{\mathbf A, \mathbf M}^0)+(h-h_0)\|_{\mathbf M})+\varepsilon\nonumber\\
& \le &  \gamma_1 \|h_0\|_{{\mathbf H}_1}+
\gamma_1 \|h_1\|_{{\mathbf H}_1}+(\gamma_1+\gamma_2)  \sqrt{a_{\mathbf A}} \|h-h_0\|_{\mathbf M}\nonumber\\
& & + 2 \gamma_2  \sqrt{a_{\mathbf A}}\sigma_{\mathbf A, \mathbf M}(x^0)+\varepsilon\nonumber\\
& \le & \gamma_1  \|h_0\|_{{\mathbf H}_1}+
\gamma_1 \|h_1\|_{{\mathbf H}_1}+ (\gamma_1+\gamma_2) \sqrt{a_{\mathbf A}} \|h_0\|_{\mathbf M}\nonumber\\
& & +
 2(\gamma_1+2\gamma_2) \sqrt{a_{\mathbf A}} \sigma_{\mathbf A, \mathbf M}(x^0)+\varepsilon.
\end{eqnarray}

By the definition of $x_{{\bf A}, {\bf M}}^0$, we have that
$$
x_{{\bf A}, {\bf M}}^0= P_{{\bf A}(x_{{\bf A}, {\bf M}}^0)}(x^0) \ {\rm and}
\ x_{{\bf A}, {\bf M}}^0+h_0=P_{{\bf A}(x_{{\bf A}, {\bf M}}^0)}(x^0_{\bf M}),$$
where $P_{\mathbf V}$ is the projection operator from ${\mathbf H}_1$ to its closed subspace ${\mathbf V}$.
Therefore 
\begin{equation}\label{h0projection}
h_0=P_{{\bf A}(x_{{\bf A}, {\bf M}}^0)}(h).\end{equation}
By \eqref{sap.def}, \eqref{optimization.thm1.pf.eq0}, 
\eqref{optimization.thm1.pf.eq2b-} and \eqref{h0projection} we get
\begin{eqnarray} \label{optimization.thm1.pf.eq6}
\|h_0\|_{{\mathbf H}_1} & = & \|P_{{\mathbf A}(x_{\mathbf A, \mathbf M}^0)}(h)\|_{{\mathbf H}_1}\le \|h\|_{{\mathbf H}_1}\nonumber\\
&\le & D \|F(h)\|
+\beta \sqrt{a_{\mathbf A}} \sigma_{\mathbf A,\mathbf M}(h)\nonumber\\
& \le &
D \|F(h)\|+\beta \sqrt{a_{\mathbf A}} \|h-h_0\|_{\mathbf M}
\end{eqnarray}
and
\begin{eqnarray}  \label{optimization.thm1.pf.eq7}
\|h_1\|_{{\mathbf H}_1} & = & \|P_{{\mathbf A}(h_1)}(I-P_{{\mathbf A}(x_{\mathbf A, \mathbf M}^0)})(h)\|_{{\mathbf H}_1}\le
\|h\|_{{\mathbf H}_1} 
\nonumber\\
 & \le &  D \|F(h)\| +\beta \sqrt{a_{\mathbf A}} \|h-h_0\|_{\mathbf M}.
\end{eqnarray}
 Hence by \eqref{rationB1M}, \eqref{optimization.thm1.pf.eq4}, \eqref{optimization.thm1.pf.eq5}, \eqref{optimization.thm1.pf.eq6}
and  \eqref{optimization.thm1.pf.eq7}, we have
\begin{eqnarray} \label{optimization.thm1.pf.eq8}
\|h_0\|_{{\mathbf  H}_1} & \le &
D\gamma_1 \|h_0\|_{{\mathbf H}_1}+(D\gamma_1+D\gamma_2+\beta) \sqrt{a_{\mathbf A}}\|h_0\|_{\mathbf M}+
D \gamma_1 \|h_1\|_{{\mathbf H}_1}\nonumber\\
& &
+ 2(D\gamma_1+2D\gamma_2+\beta) \sqrt{a_{\mathbf A}}\sigma_{\mathbf A, \mathbf M}(x^0)
+ D\varepsilon; 
\end{eqnarray}
and
\begin{eqnarray} \label{optimization.thm1.pf.eq9}
\|h_1\|_{{\mathbf H}_1} & \le &
D\gamma_1\|h_0\|_{{\mathbf H}_1}+(D\gamma_1+D\gamma_2+\beta) \sqrt{a_{\mathbf A}}\|h_0\|_{\mathbf M}+
D \gamma_1 \|h_1\|_{{\mathbf H}_1}\nonumber\\
& &
+ 2(D\gamma_1+2D\gamma_2+\beta) \sqrt{a_{\mathbf A}}\sigma_{\mathbf A, \mathbf M}(x^0)
+ D\varepsilon.
\end{eqnarray}

Combining \eqref{optimization.thm1.pf.eq8} and \eqref{optimization.thm1.pf.eq9} and using $\|h_0\|_{\mathbf M}\le \sqrt{s_{\mathbf A}}\|h_0\|_{{\mathbf H}_1}$ lead to
\begin{equation}  \label{optimization.thm1.pf.eq10}
\|h_0\|_{{\mathbf H}_1}\le
\frac{2(D\gamma_1+2D\gamma_2+\beta) \sqrt{a_{\mathbf A}}\sigma_{\mathbf A, \mathbf M}(x^0) + D\varepsilon}
{1-2D\gamma_1-(D\gamma_1+D\gamma_2+\beta)\sqrt{a_{\mathbf A}s_{\mathbf A}}}
\end{equation}
and
\begin{equation}  \label{optimization.thm1.pf.eq11}
\|h_1\|_{{\mathbf H}_1}\le
\frac{2(D\gamma_1+2D\gamma_2+\beta) \sqrt{a_{\mathbf A}}\sigma_{\mathbf A, \mathbf M}(x^0) + D\varepsilon}
{1-2D\gamma_1-(D\gamma_1+D\gamma_2+\beta)\sqrt{a_{\mathbf A}s_{\mathbf A}}}.
\end{equation}
On the other hand,
\begin{eqnarray} \label{optimization.thm1.pf.eq12}
\|h-h_0-h_1\|_{{\mathbf H}_1}
 & \le &  \sqrt{a_{\mathbf A}} \|h-h_0\|_{\mathbf M}\le
\sqrt{a_{\mathbf A}} \|h_0\|_{\mathbf M}+ 2 \sqrt{a_{\mathbf A}}\sigma_{\mathbf A, \mathbf M}(x^0)\nonumber\\
& \le & \sqrt{a_{\mathbf A}s_{\mathbf A}}\|h_0\|_{{\mathbf H}_1}+ 2\sqrt{a_{\mathbf A}}\sigma_{\mathbf A, \mathbf M}(x^0)
\end{eqnarray}
by \eqref{rationB1M},
\eqref{optimization.thm1.pf.eq4} and
\eqref{optimization.thm1.pf.eq2c}.
Therefore
the  error estimates \eqref{optimization.thm1.error} and \eqref{optimization.thm1.error2}  follow from
 \eqref{rationB1M}, \eqref{optimization.thm1.pf.eq4},
 \eqref{optimization.thm1.pf.eq10}, \eqref{optimization.thm1.pf.eq11}
and \eqref{optimization.thm1.pf.eq12}.
\end{proof}

As a corollary, we have the following result for linear mapping $F$, cf. \cite[Theorem 1.1]{sunieee11}
in the classical sparse recovery setting.

\begin{cor}\label{optimization.cor1}
Let $ {\mathbf M},
{\mathbf A}, {\mathbf H}_1, {\mathbf H}_2$ be as in Theorem \ref{optimization.thm1}, and
let  $F: {\mathbf H}_1\longmapsto {\mathbf H}_2$ be  linear
 and have the sparse Riesz property
\eqref{sap.def}  with $D>0$ and $\beta\in (0,1/\sqrt{a_{\mathbf A}s_{\mathbf A}})$.
Given $x^0\in {\mathbf M}$ and $\varepsilon>0$,
the optimization solution of \eqref{optimizationsolution}
 satisfies
$$
\|x_{\mathbf M}^0-x^0\|_{{\mathbf H}_1} \le
\Big(   \frac{2+4\beta}
{1-\beta\sqrt{a_{\mathbf A}s_{\mathbf A}}} \Big)
\sqrt{a_{\mathbf A}}\sigma_{\mathbf A, \mathbf M}(x^0) +\
\frac{(2+\sqrt{a_{\mathbf A}s_{\mathbf A}})D}
{1-\beta\sqrt{a_{\mathbf A}s_{\mathbf A}}}
 \varepsilon
$$
and
$$
\|x_{\mathbf M}^0-x^0\|_{\mathbf M} \le
\Big( \frac{2+2 \beta\sqrt{a_{\mathbf A}s_{\mathbf A}}}
{1-\beta\sqrt{a_{\mathbf A}s_{\mathbf A}}}\Big)
\sigma_{\mathbf A, \mathbf M}(x^0) +\
\frac{2D}
{1-\beta\sqrt{a_{\mathbf A}s_{\mathbf A}}}
\sqrt{s_{\mathbf A}}
 \varepsilon,
$$
where $\sigma_{\mathbf A, \mathbf M}(x^0)=\inf_{\hat x\in {\mathbf A}}\|\hat x-x^0\|_{\mathbf M}$.
\end{cor}

\section{Sparse Riesz property and almost linear property}
\label{sparserieszproperty.section}

In this section, we consider  the sparse Riesz property
\eqref{sap.def} and almost linear property \eqref{almostlinearproperty} for nonlinear maps
 not far from a linear operator with the restricted isometry property \eqref{sparsealgorithm.thm.eq2}.

\begin{thm}\label{sparsereconstruction.optimization.thm}
Let ${\mathbf H}_1$  and ${\mathbf H}_2$ be Hilbert spaces, ${\mathbf M}$ be a Banach space,
 and ${\mathbf A}=\cup_{i\in I} {\mathbf A}_i$ be a union of closed linear subspaces
 of ${\mathbf H}_1$.
Assume that
 $({\mathbf A}, {\mathbf M}, {\mathbf H}_1)$ is a sparse approximation triple,
and
 $T\in {\mathcal B}({\mathbf H}_1, {\mathbf H}_2)$ has
 the
 restricted isometry property  \eqref{sparsealgorithm.thm.eq2} on $2{\mathbf A}$ with
 $\delta_{2{\mathbf A}}(T)<\sqrt{2}/2$.
If  $F$ is a continuous map from ${\mathbf H}_1$ to ${\mathbf H}_2$
with $F(0)=0$ and
\begin{equation} \gamma_{F, T}(2{\bf A})< \frac{\sqrt{2}}{2}-\sqrt{\delta_{2{\bf A}}(T)},
\end{equation}
then
 $F$ has the sparse Riesz property \eqref{sap.def},
\begin{eqnarray}\label{Dbeta}
\|F(x)\|_{{\mathbf H}_2} & \ge &  \big(1-\sqrt{2} (\sqrt{\delta_{2{\mathbf A}}(T)}+ \gamma_{F,T}(2{\mathbf A}))\big)
\|x\|_{{\mathbf H}_1}\nonumber\\
& & - \big(\sqrt{\delta_{2{\mathbf A}}(T)}+ \gamma_{F,T}(2{\mathbf A})\big)
\sqrt{a_{\bf A}} \sigma_{{\bf A}, {\bf M}}(x),\  x\in M.
\end{eqnarray}
\end{thm}

For any $x\in {\mathbf M}$, define
\begin{equation}\label{greedyalgorithm.m}
x_{\mathbf A, \mathbf M}^{k+1}=x_{\mathbf A, \mathbf M}^k+ {\rm argmin}_{\hat x\in {\mathbf A}}\| x-x_{\mathbf A, \mathbf M}^k-\hat x\|_{\mathbf M},\  k\ge 0, \end{equation}
with initial $x_{\mathbf A, \mathbf M}^0=0$. 
To prove Theorem \ref{sparsereconstruction.optimization.thm}, we need
convergence of the above greedy algorithm.

\begin{prop}\label{greedy.prop}
Let  $({\mathbf A}, {\mathbf M}, {\mathbf H}_1)$ be a sparse approximation triple.
Then $x_{\mathbf A, \mathbf M}^k, k\ge 0$, in the greedy algorithm \eqref{greedyalgorithm.m}
converges to $x\in {\mathbf M}$,
 \begin{equation}\label{greedylimit}
\lim_{k\to \infty} \|x_{\mathbf A, \mathbf M}^k-x\|_{\mathbf M}=0.\end{equation}
\end{prop}

We postpone the proof of the above proposition to Appendix B and start the
proof of Theorem \ref{sparsereconstruction.optimization.thm}.

\begin{proof}
Take $x\in {\mathbf M}$, let $x_{\mathbf A, \mathbf M}^k, k\ge 0$,
be as in the greedy algorithm \eqref{greedyalgorithm.m}. Then
from Proposition \ref{greedy.prop}, 
 the continuity of $F$ on ${\mathbf H}_1$, and the continuous imbedding
of  ${\mathbf M}$ into ${\mathbf H}_1$ it follows that
\begin{equation} \label{sparsereconstruction.optimization.thm2.pf.eq0+}
\lim_{k\to \infty} \|F(x_{\mathbf A, \mathbf M}^k)-F(x)\|_{{\bf H}_2}=0.
\end{equation}
Write
$u_k=x_{\mathbf A, \mathbf M}^{k+1}-x_{\mathbf A, \mathbf M}^{k}, k\ge 0$.
Then
 $u_k\in {\mathbf A}$, and
\begin{eqnarray} \label{sparsereconstruction.optimization.thm2.pf.eq1}
\|F(x)-Tx\|_{{\bf H}_2}
& \le &  \sum_{k=0}^\infty
\big\|F(x_{\mathbf A, \mathbf M}^{k+1})-F(x_{\mathbf A, \mathbf M}^k)-T u_k\big\|_{{\bf H}_2}\nonumber\\
&\le &  \gamma_{F,T}(2{\mathbf A}) \sum_{k=0}^\infty \|u_k\|_{{\mathbf H}_1}
\end{eqnarray}
by \eqref{sparsealgorithm.thm.eq2}, \eqref{tildegammafta},
\eqref{sparsereconstruction.optimization.thm2.pf.eq0+} and the assumption $F(0)=0$.

Observe that
$4\langle T\tilde u_k, T \tilde u_{k'}\rangle = \|T(\tilde u_k+\tilde u_{k'})\|_{{\mathbf H}_2}^2-\|T(\tilde u_k-\tilde u_{k'})\|_{{\mathbf H}_2}^2$,
where $\tilde u_k= u_k/\|u_k\|_{{\mathbf H}_1}, k\ge 0$.
Then
\begin{equation}\label{sparsereconstruction.optimization.thm2.pf.eq2}
\big|\langle Tu_k, T u_{k'}\rangle-\langle u_k, u_{k'}\rangle|
\le \delta_{2{\mathbf A}}(T) \|u_k\|_{{\mathbf H}_1}\|u_{k'}\|_{{\mathbf H}_1}, \ k, k'\ge 0,
\end{equation}
by the restricted isometry property \eqref{sparsealgorithm.thm.eq2}.
We remark that in the classical sparse recovery setting,
the inner product $\langle u_k, u_{k'}\rangle$ between different $u_{k}$ and $u_{k'}$ is always zero, but it may be nonzero in our  setting.
Hence for $K\ge 1$,
\begin{eqnarray} \label{sparsereconstruction.optimization.thm2.pf.eq3}
\big\|Tx_{\mathbf A, \mathbf M}^{K+1}\big\|_{{\bf H}_2}^2 & =
& \Big\|T\Big(\sum_{k=0}^K u_k\Big)\Big\|_{{\bf H}_2}^2
=  \sum_{k=0}^K \|Tu_k\|_{{\bf H}_2}^2+
\sum_{0\le k\ne k'\le K} \langle Tu_k, Tu_{k'}\rangle\nonumber\\
& \le &  (1+\delta_{2{\mathbf A}}(T))
\sum_{k=0}^K \|u_k\|_{{\mathbf H}_1}^2
+
\sum_{0\le k\ne k'\le K} \langle u_k, u_{k'}\rangle\nonumber\\
& & +\delta_{2{\mathbf A}}(T) \sum_{0\le k\ne k'\le K} \|u_k\|_{{\mathbf H}_1}\|u_{k'}\|_{{\mathbf H}_1}\nonumber\\
& = & 
\big\|x_{\mathbf A, \mathbf M}^{K+1}\big\|_{{\mathbf H}_1}^2+
\delta_{2{\mathbf A}}(T) \Big(\sum_{k=0}^K \|u_{k}\|_{{\mathbf H}_1}\Big)^2,
\end{eqnarray}
and similarly
\begin{equation} \label{sparsereconstruction.optimization.thm2.pf.eq4}
\|Tx_{\mathbf A, \mathbf M}^{K+1}\|_{{\bf H}_2}^2 
 \ge  
 \big\|x_{\mathbf A, \mathbf M}^{K+1}\big\|_{{\mathbf H}_1}^2-
\delta_{2{\mathbf A}}(T) \Big(\sum_{k=0}^K \|u_{k}\|_{{\mathbf H}_1}\Big)^2.
\end{equation}
Therefore combining
\eqref{sparsereconstruction.optimization.thm2.pf.eq3}
and
\eqref{sparsereconstruction.optimization.thm2.pf.eq4},
and then applying  \eqref{optimization.thm1.eq1} and \eqref{greedylimit}
  when  taking limit  as $K\to \infty$, we obtain
\begin{equation*}
 -\delta_{2{\mathbf A}}(T)
 \Big(\sum_{k\ge 0} \|u_{k}\|_{{\mathbf H}_1}\Big)^2\le
\|Tx\|_{{\mathbf H}_2}^2-\|x\|_{{\mathbf H}_1}^2 \le \delta_{2{\mathbf A}}(T)
\Big(\sum_{k\ge 0} \|u_{k}\|_{{\mathbf H}_1}\Big)^2,
\end{equation*}
which implies that
\begin{equation}
 \label{sparsereconstruction.optimization.thm2.pf.eq6}
 \|Tx\|_{{\mathbf H}_2} -\sqrt{\delta_{2{\mathbf A}}(T)}
\sum_{k\ge 0} \|u_{k}\|_{{\mathbf H}_1}\le
\|x\|_{{\mathbf H}_1} \le
\|Tx\|_{{\mathbf H}_2} +\sqrt{\delta_{2{\mathbf A}}(T)}
\sum_{k\ge 0} \|u_{k}\|_{{\mathbf H}_1}.
\end{equation}

By \eqref{secondapproximationestimate},  
\begin{equation}\label{ukmh.estimate}
\|u_k\|_{{\mathbf H}_1}\le \sqrt{ a_{\mathbf A}} \|u_{k-1}\|_{\mathbf M},\  k\ge 1.
\end{equation}
This together with
\eqref{bestapproximation.hm} and \eqref{optimization.thm1.eq3}
implies that
\begin{equation}  \label{sparsereconstruction.optimization.thm2.pf.eq7}
\sum_{k\ge 0} \|u_{k}\|_{{\mathbf H}_1}  \le  \|u_0\|_{{\mathbf H}_1}+ \|u_1\|_{{\mathbf H}_1}+ \sqrt{ a_{\mathbf A}}
\sum_{k\ge 2} \|u_{k-1}\|_{\mathbf M}\le \sqrt{2} \|x\|_{{\mathbf H}_1}+ \sqrt{ a_{\mathbf A}} \sigma_{{\mathbf A}, {\mathbf M}}(x).
\end{equation}
Combining
 \eqref{sparsereconstruction.optimization.thm2.pf.eq1},
  \eqref{sparsereconstruction.optimization.thm2.pf.eq6}
and
  \eqref{sparsereconstruction.optimization.thm2.pf.eq7}
  gives
  \begin{equation*}
\big|\|F(x)\|_{{\mathbf H}_2}-\|x\|_{{\mathbf H}_1}\big|\le
 \big (\sqrt{\delta_{2{\mathbf A}}(T)}+
  \gamma_{F,T}(2{\mathbf A})\big)
 \big(\sqrt{2}\|x\|_{{\mathbf H}_1}+ \sqrt{a_{\mathbf A}}  \sigma_{{\mathbf A}, {\mathbf M}}(x)\big).
\end{equation*}
Reformulating the above estimates completes the proof of the estimate
\eqref{Dbeta} for the sparse Riesz property of $F$.
\end{proof}

\begin{thm}\label{almostlinear.thm}
Let ${\mathbf H}_1, {\mathbf H}_2, {\mathbf M}, {\mathbf A}, T, F$
be as in Theorem \ref{sparsereconstruction.optimization.thm} with additional assumption that $\gamma_{F,T}(4{\mathbf A})<\infty$.
Then
$F$ has the almost linear property on ${\bf A}$,
  \begin{eqnarray}\label{gamma1gamma2} 
 & & \|F(x)-F(y)-F(x-y)\|_{{\mathbf H}_2}  \le   2 \gamma_{F,T}(4{\mathbf A})
 \|x-y\|_{{\mathbf H}_1}\nonumber\\
 & &\qquad\qquad  +
 2\big( \gamma_{F,T}(2{\mathbf A})+\gamma_{F,T}(4{\mathbf A})\big)
\sqrt{a_{\bf A}}( \sigma_{\mathbf A, \mathbf M}(x)+\sigma_{\mathbf A, \mathbf M}(y)).
\end{eqnarray}
 \end{thm}

\begin{proof} Take $x, y\in {\bf M}$,  and let $x^k_{{\bf A}, {\bf M}}$ and $y^k_{{\bf A}, {\bf M}}, k\ge 0$, be as
in the greedy algorithm \eqref{greedyalgorithm.m} to approximate $x$ and $y\in {\bf M}$ respectively.
Write
\begin{eqnarray}\label{alpeq.pf.eq1}
 \|F(x)-F(y)-T(x-y)\|_{{\mathbf H}_2}
 & \le &
\|F(x)-F(x_{\mathbf A, \mathbf M}^2)-T(x-x_{\mathbf A, \mathbf M}^2)\|_{{\mathbf H}_2}\nonumber\\
& & +
\|F(y)-F(y_{\mathbf A, \mathbf M}^2)-T(y-y_{\mathbf A, \mathbf M}^2)\|_{{\mathbf H}_2}\nonumber\\
& & + \|F(x_{\mathbf A, \mathbf M}^2)-F(y_{\mathbf A, \mathbf M}^2)-T(x_{\mathbf A, \mathbf M}^2-y_{\mathbf A, \mathbf M}^2)\|_{{\mathbf H}_2}\nonumber\\
& =: & I_1+I_2+I_3.
\end{eqnarray}
By \eqref{secondapproximationestimate}, \eqref{sparsealgorithm.thm.eq2}, \eqref{tildegammafta},
 \eqref{greedylimit} 
  and the continuity of $F$  and $T$ on ${\mathbf H}_1$, we get
\begin{eqnarray}\label{alpeq.pf.eq2}
 I_1
& \le &
\sum_{k\ge 2}
\big\|F(x_{\mathbf A, \mathbf M}^{k+1})-F(x_{\mathbf A, \mathbf M}^k)-T ( x_{\mathbf A, \mathbf M}^{k+1}-x_{\mathbf A, \mathbf M}^k)\big\|_{{\mathbf H}_2}\nonumber\\
& \le &
 \gamma_{F,T}({\mathbf A})
\sum_{k\ge 2} \| x_{\mathbf A, \mathbf M}^{k+1}-x_{\mathbf A, \mathbf M}^{k}\|_{{\mathbf H}_1}\nonumber\\
&\le &
 \gamma_{F,T}({\mathbf A}) \sqrt{a_{\mathbf A}}
\sigma_{\mathbf A, \mathbf M}(x)\le  \gamma_{F,T}(2{\mathbf A}) \sqrt{a_{\mathbf A}}
\sigma_{\mathbf A, \mathbf M}(x)
\end{eqnarray}
and similarly
\begin{equation} \label{alpeq.pf.eq3}
 I_2
\le
 \gamma_{F,T}(2{\mathbf A}) \sqrt{a_{\mathbf A}} \sigma_{\mathbf A, \mathbf M}(y).
\end{equation}
For the term $I_3$, we obtain from \eqref{secondapproximationestimate}
and \eqref{tildegammafta}
that
\begin{eqnarray} \label{alpeq.pf.eq4}
 I_3 &\le  &
 \gamma_{F,T}(4{\mathbf A})  \|x_{\mathbf A, \mathbf M}^2- y_{\mathbf A, \mathbf M}^2\|_{{\mathbf H}_1}\nonumber\\
& \le &
 \gamma_{F,T}(4{\mathbf A})  (\|x-y\|_{{\mathbf H}_1} +\|x-x_{\mathbf A, \mathbf M}^2\|_{{\mathbf H}_1}+
 \|y-y_{\mathbf A, \mathbf M}^2\|_{{\mathbf H}_1}\big)\nonumber\\
 & \le &  \gamma_{F,T}(4{\mathbf A}) \big ( \|x-y\|_{{\mathbf H}_1}+
 \sqrt{a_{\mathbf A}} \big(\sigma_{\mathbf A, \mathbf M}(x)+\sigma_{\mathbf A, \mathbf M}(y)\big)\big).
\end{eqnarray}
Combining estimates in \eqref{alpeq.pf.eq1}--\eqref{alpeq.pf.eq4} gives
\begin{eqnarray}\label {alpeq.pf.eq5}
\|F(x)-F(y)-T(x-y)\| & \le &
\gamma_{F,T}(4{\mathbf A}) \|x-y\|_{{\mathbf H}_1} \nonumber\\
& &
+ \big( \gamma_{F,T}(2{\mathbf A})+\gamma_{F,T}(4{\mathbf A})\big) \sqrt{a_{\bf A}}\big(\sigma_{\mathbf A, \mathbf M}(x)+\sigma_{\mathbf A, \mathbf M}(y)
 \big).
 \end{eqnarray}

 Write
 \begin{eqnarray*} \label{alpeq.pf.eq6}
 & & \|F(x-y)-T(x-y)\|\nonumber\\
 & \le &
 \|F(x-y)-F(x^2_{\mathbf A, \mathbf M}-y)-T(x-x^2_{\mathbf A, \mathbf M})\|
\nonumber\\
& & +  \|F(x^2_{\mathbf A, \mathbf M}-y)-F(x^2_{\mathbf A, \mathbf M}-y^2_{\mathbf A, \mathbf M})-T(y^2_{\mathbf A, \mathbf M}-y)\|\nonumber\\
& &+ \|F(x^2_{\mathbf A, \mathbf M}-y^2_{\mathbf A, \mathbf M})-F(0)- T(x^2_{\mathbf A, \mathbf M}-y^2_{\mathbf A, \mathbf M})\|.
 \end{eqnarray*}
 Following the arguments used to establish \eqref{alpeq.pf.eq5}, we have
\begin{eqnarray}\label {alpeq.pf.eq7}
  \|F(x-y)-T(x-y)\| &\le &
  \gamma_{F,T}(4{\mathbf A}) \|x-y\|_{{\mathbf H}_1} \nonumber\\
& &
+ \big( \gamma_{F,T}(2{\mathbf A})+\gamma_{F,T}(4{\mathbf A})\big) \sqrt{a_{\bf A}}\big(\sigma_{\mathbf A, \mathbf M}(x)+\sigma_{\mathbf A, \mathbf M}(y)
 \big).
 \end{eqnarray}
 Combining \eqref{alpeq.pf.eq5} and \eqref{alpeq.pf.eq7} proves
 the estimate \eqref{gamma1gamma2} for
  the almost linear property of the map $F$.
  \end{proof}

\smallskip
Combining Theorems \ref{optimization.thm1}, \ref{sparsereconstruction.optimization.thm} and
\ref{almostlinear.thm} leads to the following result on the stable reconstruction of
sparse signals  $x$ from their nonlinear measurements $F(x)$ when $F$ is
not far away from a measurement matrix $T$ with the restricted isometry property \eqref{sparsealgorithm.thm.eq2}.

\begin{thm}\label{optimization.thm2}
Let ${\mathbf H}_1, {\mathbf M}, {\mathbf A}, T, F$
be as in Theorem \ref{sparsereconstruction.optimization.thm}
with
 \begin{eqnarray*}
 & & \sqrt{2} \big(\sqrt{\delta_{2{\mathbf A}}(T)}+ \gamma_{F,T}(2{\mathbf A})\big)+ 4
\gamma_{F,T}(4{\mathbf A})\\
& & \qquad + (\sqrt{\delta_{2{\mathbf A}}(T)}+ 3 \gamma_{F,T}(2{\mathbf A})+ 4\sqrt{\delta_{4{\mathbf A}}(T)}\big) \sqrt{a_{\bf A} s_{\bf A}}<1.
\end{eqnarray*}
Then  for any given $x^0\in {\mathbf M}$
and $\varepsilon>0$, the solution $x_{\mathbf M}^0$ of the minimization problem \eqref{optimizationsolution}
has the following error estimates:
\begin{equation}\label{optimization.thm2.error1}
\|x_{\mathbf M}^0-x^0\|_{{\mathbf H}_1}\le C_1 \sqrt{a_{\mathbf A}} \sigma_{\mathbf A, \mathbf M}(x^0)+C_2\epsilon\end{equation}
and
\begin{equation}
\label{optimization.thm2.error2}\|x_{\mathbf M}^0-x^0\|_{\mathbf M}\le C_1 \sigma_{\mathbf A, \mathbf M}(x^0)+C_2\sqrt{s_{\mathbf A}}\epsilon\end{equation}
where $C_1$ and $C_2$ are  absolute constants  independent on $x^0\in {\bf M}$ and $\epsilon\ge 0$.
\end{thm}

Applying Theorem \ref{optimization.thm2} to linear maps,
 we have the following corollary.

 \begin{cor} \label{optimization.thm2.cor}
 Let ${\mathbf H}_1, {\mathbf M}, {\mathbf A}, {\mathbf H}_2$ and
  $T$ be as in Theorem \ref{optimization.thm2}.
  If
$$  \delta_{2{\mathbf A}}(T)< (
 \sqrt{2} +
 \sqrt{a_{\bf A} s_{\bf A}})^{-2},
$$
then  for any given $x^0\in {\mathbf M}$
and $\varepsilon>0$, the solution $x_{\mathbf M}^0$ of the minimization problem \eqref{optimizationsolution}
with $F=T$
has the error  estimates
\eqref{optimization.thm2.error1}
and \eqref{optimization.thm2.error2}.
\end{cor}

For classical sparse recovery problems, the conclusions in Corollary \ref{optimization.thm2.cor}
have been established under weaker assumptions on the restricted isometry constant $\delta_{2{\mathbf A}}(T)$,
see \cite{czacha13} 
 and references therein.

\bigskip
{\bf Acknowledgement}\quad
The first author thanks Professor Yuesheng Xu for his invitation to visit Guangdong Province Key Laboratory of Computational Science at
 Sun Yat-sen University, China, where part of this work was done.

\begin{appendix}
\section{
Bi-Lipschitz map and uniform stability} 
\label{bilipschitz.appendix}

In this appendix, we 
 provide some  sufficient conditions, mostly optimal,  for a differentiable map to have
the bi-Lipschitz property \eqref{bilipschitzmap.def}, see  Theorems \ref{frame.sufficientcondition} and \ref{frame.sufficientcondition.linear}
  in Banach space setting,  and Theorems \ref{frame.sufficientcondition.hilbert} and \ref{frame.sufficientcondition.hilberts2}  in Hilbert space setting.

\smallskip

For a differentiable   map $F$ from one Banach space ${\mathbf B}_1$ to another Banach space ${\mathbf B}_2$
that has the bi-Lipschitz property \eqref{bilipschitzmap.def},
we have
\begin{equation*}
A   \|y\|\le \frac{\|F(x+ty)-F(x)\|}{t} \le B  \|y\|\quad {\rm for \ all} \ \  x, y\in {\mathbf B}_1\ {\rm and} \ t>0,
\end{equation*}
where $A,B$ are the constants in the bi-Lipschitz property \eqref{bilipschitzmap.def}.
Then taking limit as $t\to 0$ 
leads to a necessary condition for a differentiable bi-Lipschitz  map.

\begin{thm} \label{bilipschitz.thm} Let  ${\mathbf B}_1$ and ${\mathbf B}_2$ be Banach spaces. If
$F:{\mathbf B}_1\to {\mathbf B}_2$ is a differentiable  map that has the bi-Lipschitz property \eqref{bilipschitzmap.def}, then its derivative
  $F'(x), x\in {\mathbf B}_1$, has the uniform stability property \eqref{localstability.thm.eq1}.
\end{thm}

 For ${\bf B}_1={\bf B}_2=\RR$, a differentiable map $F$ with the uniform stability property \eqref{localstability.thm.eq1} for its derivative
has the bi-Lipschitz  property \eqref{bilipschitzmap.def}, but it is not true in general Banach space setting.
Maps $E_{p, \epsilon}, 1\le p\le \infty, \epsilon \in [0, \pi/4)$, from $\RR$ to $\RR^2$
in the example below are such examples.

\begin{ex}\label{epepsilon.example}
 {\rm  For $1\le p\le \infty$ and $\epsilon\in [0, \pi/4)$, define 
 $E_{p, \epsilon}: \RR\longmapsto \RR^2$  by
\begin{equation}\label{exponentialmap.def}
E_{p,, \epsilon}(t)=\left\{\begin{array}{l} (-\cos \epsilon, \sin \epsilon)- (\sin \epsilon, \cos \epsilon)(t+\pi/2+\epsilon)\\
\hfill {\rm if} \ t\in (-\infty, -\pi/2-\epsilon),\\
(\sin t, -\cos t) \hfill {\rm if} \ t\in [-\pi/2-\epsilon, \pi/2+\epsilon],\\
(\cos \epsilon, \sin \epsilon)+ (-\sin \epsilon, \cos \epsilon)(t-\pi/2-\epsilon)\\
\hfill {\rm if} \ t\in (\pi/2+\epsilon, \infty),
\end{array}\right.
\end{equation}
see Figure \ref{functionepepsilon.fig}.
\begin{figure}[hbt]
\centering
\begin{tabular}{cc}
  \includegraphics [width=58mm, height=58mm]{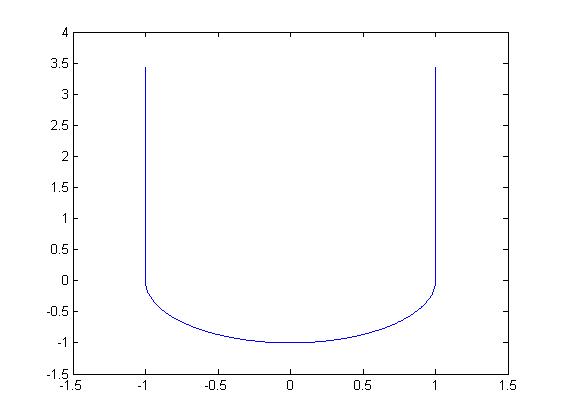} & \includegraphics [width=58mm, height=58mm]{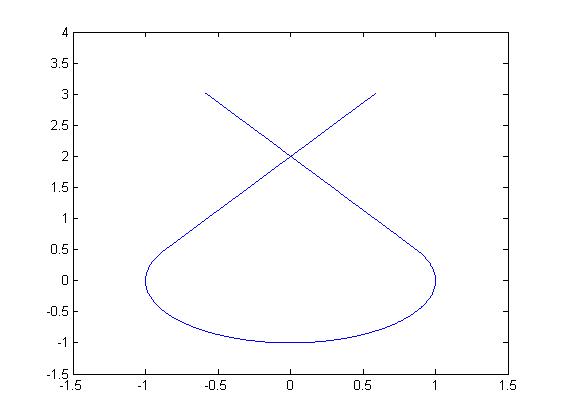}
   \end{tabular}
   \caption{\small Maps $E_{p, \epsilon}$ from $\RR$ to $\RR^2$ with $\epsilon=0$ (left) and $\epsilon=\pi/6$ (right).}
\label{functionepepsilon.fig}
\end{figure}
The maps $E_{p, \epsilon}$ just defined do not have the bi-Lipschitz property \eqref{bilipschitzmap.def}, but their derivatives
$E_{p, \epsilon}^\prime$ 
have the uniform stability property
\eqref{localstability.thm.eq1},
\begin{eqnarray*} \frac{\sqrt{2}}{2} |\tilde t| & \le  & \|E_{p, \epsilon}^\prime(t) \tilde t\|_p=\left\{\begin{array}{ll}
\|(\tilde t\sin \epsilon, \tilde t\cos \epsilon)\|_p & {\rm if} \ t< -\pi/2-\epsilon\\
 \|(\tilde t\cos t, \tilde t\sin t)\|_p & {\rm if} \ |t|\le \pi/2+\epsilon\\
\|(-\tilde t \sin \epsilon, \tilde t \cos \epsilon)\|_p & {\rm if} \ t>\pi/2+\epsilon\end{array}\right.\nonumber\\
&\le &  2|\tilde t|\ \  {\rm for \ all} \  t,\tilde t\in \RR,
\end{eqnarray*}
where $\|\cdot\|_p, 1\le p\le \infty$, is the $p$-norm on the Euclidean space $\RR^2$.
}\end{ex}


 Given a differentiable bi-Lipschitz map $F$ from one Banach space ${\mathbf B}_1$ to another
Banach space  ${\mathbf B}_2$ such that  its derivative $F'(x)$ is uniformly stable, define
\begin{equation}\label{alphaF.def}
\alpha_F:=\sup_{\|y\|=1}\inf_{\|z\|=1}\sup_{x\in {\mathbf B}_1}\Big\|\frac{F'(x)y}{\|F'(x)y\|}-z\Big\|.
\end{equation}
The quantity $\alpha_F$ is the minimal radius such that
  for any $0\ne y\in {\mathbf B}_1$, the set ${\mathbb B}(y)$ of
unit vectors $F'(x)y/\|F'(x)y\|, x\in {\mathbf B}_1$,
is contained in  a ball of radius $\alpha_F<1$ centered at a unit vector.
Our next theorem shows that
 a differentiable bi-Lipschitz map $F$ with its derivative $F'(x)$ being uniformly stable and continuous
and with  $\alpha_F$ in \eqref{alphaF.def} satisfying $\alpha_F<1$ has the bi-Lipschitz property \eqref{bilipschitzmap.def}.

\begin{thm} \label{frame.sufficientcondition}
Let  ${\mathbf B}_1$ and ${\mathbf B}_2$ be Banach spaces, and
$F$ be a continuously differentiable map from ${\mathbf B}_1$ to ${\mathbf B}_2$
 with the property that  its derivative has the uniform stability property
\eqref {localstability.thm.eq1}.
If $\alpha_F$ in \eqref{alphaF.def} satisfies
\begin{equation}\label{ballcondition}
\alpha_F<1,\end{equation}
then $F$ is a bi-Lipschitz map.
\end{thm}

\begin{proof} Given $x, y\in {\mathbf B}_1$ with $y\ne 0$,
\begin{eqnarray*}
F(x+y)-F(x) & = &\int_0^1 F'(x+ty) y dt
=  \Big (\int_0^1 \|F'(x+ty)y\|dt\Big) z\\
& & +
\int_0^1 \|F'(x+ty)y\| \Big(\frac{F'(x+ty)y}{\|F'(x+ty)y\|}-z\Big)dt,\end{eqnarray*}
where $z\in {\mathbf B}_2$ with $\|z\|=1$.
Thus
\begin{eqnarray*}
\|F(x+y)-F(x)\| & \ge &  \Big(\int_0^1 \|F'(x+ty)y\|dt\Big) \\
& & \quad \times \Big(1-\inf_{\|z\|=1}
\sup_{0\le t\le 1}\Big\|\frac{F'(x+ty)y}{\|F'(x+ty)y\|}-z\Big\|\Big)\\
& \ge &  (1-\alpha_F)
\Big(\int_0^1 \|F'(x+ty)y\|dt\Big)\ge (1-\alpha_F) A\|y\|,
\end{eqnarray*}
and
\begin{equation*}
\|F(x+y)-F(x)\|  \le
\int_0^1 \|F'(x+ty)y\|dt\le  
B\|y\|,
\end{equation*}
where $A, B$ are lower and upper stability bounds in the uniform stability property
\eqref{localstability.thm.eq1}. Combining the above two estimates completes the proof.
\end{proof}

\begin{rem}\label{alphaF.rem} {\rm
The U-shaped map $E_{p, \epsilon}$ in Example \ref{epepsilon.example}  with $p=\infty$ and $\epsilon=0$
is not a bi-Lipschitz map and     
\begin{eqnarray*} \alpha_{E_{\infty, 0}} 
& = & \inf_{\|z\|_\infty=1} \sup_{|t|\le \pi/2}\Big\|\frac{(\cos t, \sin t)}{\max(|\cos t|, |\sin t|)}-z\Big\|_\infty\\
& = & \sup_{|t|\le \pi/2}\Big\|\frac{(\cos t, \sin t)}{\max(|\cos t|, |\sin t|)}-(1, 0)\Big\|_\infty=1.\end{eqnarray*}
This indicates that the geometric condition \eqref{ballcondition} about $\alpha_F$
is optimal.
}\end{rem}

\smallskip

For  a
differentiable map $F$
not far away from a  bounded below linear operator $T$,  we  suggest
  using $Ty/\|Ty\|$ as the center of the ball containing
  the set 
  of unit vectors $F'(x) y/\|F'(x) y\|, x\in {\bf B}_1$,
 and define the minimal radius  of that ball by $\beta_{F,T}$ in \eqref{linearapproximation.center}.
Then obviously
 \begin{equation}\label{tildebeta.beta}
\alpha_F\le \beta_{F,T}.
\end{equation}
This together with Theorem \ref{frame.sufficientcondition} implies that
 a differentiable map $F$ satisfying $\beta_{F, T}<1$ is a bi-Lipschitz map.

\begin{thm} \label{frame.sufficientcondition.linear}
Let  ${\mathbf B}_1$ and ${\mathbf B}_2$ be Banach spaces, and
$F$ be a continuously differentiable map  from ${\mathbf B}_1$ to ${\mathbf B}_2$
with its derivative   having the uniform stability property
\eqref {localstability.thm.eq1}.
If $T\in {\mathcal B}({\mathbf B}_1, {\mathbf B}_2)$ is bounded below
and  satisfies \eqref{betaft.condition},
then $F$ is a bi-Lipschitz map.
\end{thm}

We may use the following quantity to measure the distance between differentiable map $F$ and bounded below linear operator $T$,
\begin{equation}\label{gammaft.def}
\delta_{F, T}:=\sup_{0\ne y\in {\mathbf B}_1} \sup_{x\in {\mathbf B}_1}\frac{\| F(x+y)-F(x)-Ty\|}{\|Ty\|}=
\sup_{0\ne y\in {\mathbf B}_1} \sup_{z\in {\mathbf B}_1}\frac{\| F'(z)y-Ty\|}{\|Ty\|}.
\end{equation}
By direct computation,
\begin{equation*} 
\beta_{F, T} 
\le   \sup_{\|y\|=1} \sup_{x\in {\mathbf B}_1} \Big(\frac{\|F'(x)y-Ty\|}{\|F'(x)y\|}+\frac{\big|\|F'(x)y\|-\|Ty\|\big|}{\|F'(x)y\|}\Big)\le \frac{2\delta_{F,T}}{1-\delta_{F,T}}.
\end{equation*}
Thus the geometric condition \eqref{betaft.condition} in
Theorem \ref{frame.sufficientcondition.linear} can be replaced by the condition
$\delta_{F,T}<1/3$. 

\begin{cor}\label{linearapproximationcondition.cor}
Let  ${\mathbf B}_1, {\mathbf B}_2, F$ and $T$ be as in Theorem \ref{frame.sufficientcondition.linear}. If
$\delta_{F,T}<1/3$, 
then $F$ is a bi-Lipschitz map.
\end{cor}

\smallskip
 The geometric condition \eqref{betaft.condition}  to guarantee the bi-Lipschitz property for the map $F$ is optimal in general Banach space setting, as
$\beta_{E_{\infty, 0}, T_1}=1$
for the U-shaped map $E_{\infty, 0}$ in Example \ref{epepsilon.example} 
 and the linear operator $T_1t:=(t,0), t\in \RR$.
But in Hilbert space setting, as shown in the next theorem, the  geometric condition \eqref{betaft.condition}
 could be relaxed to $\beta_{F,T}<\sqrt{2}$.

\begin{thm} \label{frame.sufficientcondition.hilbert}
Let  ${\mathbf H}_1$ and ${\mathbf H}_2$ be Hilbert spaces, and
let
$F:{\mathbf H}_1\to {\mathbf H}_2$
 be a continuously differentiable map with its derivative having the uniform stability property
\eqref {localstability.thm.eq1}.
If there exists a linear operator $T\in {\mathcal B}({\mathbf H}_1, {\mathbf H}_2)$
satisfying \eqref{Tstable.condition}  and \eqref{frame.sufficientcondition.hilbert.eq1},
 then $F$ is a bi-Lipschitz map.
\end{thm}

\begin{proof} Take $u, v\in {\mathbf H}_1$ with $v\ne 0$.
Then
\begin{equation} \label{frame.sufficientcondition.hilbert.pf.eq1}
\|F(u+v)-F(u)\|\le \int_0^1 \|F'(u+tv) v\|dt\le B\|v\|, 
\end{equation}
where $B$ is the upper stability bound in \eqref{localstability.thm.eq1}.
Observe that
\begin{equation*}
\langle F'(u)v, Tv\rangle= \|F'(u) v\|\|Tv\|
\Big( 1-\frac{1}{2} \Big\|\frac{F'(u)v}{\|F'(u)v\|}-\frac{Tv}{\|Tv\|}\Big\|^2\Big).
\end{equation*}
Then
\begin{equation*} \label{frame.sufficientcondition.hilbert.pf.eq2-}
\langle F'(u)v, Tv\rangle\ge \frac{2-(\beta_{F,T})^2}{2} \|F'(u) v\|\|Tv\|,  
\end{equation*}
which implies that
\begin{eqnarray} \label{frame.sufficientcondition.hilbert.pf.eq2}
\langle F(u+v)-F(u), Tv\rangle & = & \int_0^1 \langle F'(u+tv)v, Tv\rangle  dt\nonumber \\
& \ge &  \frac{2-(\beta_{F,T})^2}{2} \Big(\int_0^1 \|F'(u+tv) v\| dt\Big)\|Tv\| \nonumber\\
& \ge &  \frac{2-(\beta_{F,T})^2}{2}A \|Tv\| \|v\|,
\end{eqnarray}
where $A$ is the lower stability bound in \eqref{localstability.thm.eq1}.
Hence
\begin{equation} \label{frame.sufficientcondition.hilbert.pf.eq3}
\|F(u+v)-F(u)\|\ge \frac{\langle F(u+v)-F(u), Tv\rangle}{\|Tv\|}
\ge \frac{2-( \beta_{F,T})^2}{2}A \|v\|.
\end{equation}
Combining
\eqref{frame.sufficientcondition.hilbert.pf.eq1} and \eqref{frame.sufficientcondition.hilbert.pf.eq3}
proves the bi-Lipschitz   property for  $F$.
\end{proof}

\begin{rem}{\rm
The geometric condition \eqref{frame.sufficientcondition.hilbert.eq1} is optimal as for the U-shaped map $E_{p, \epsilon}$
in Example \ref{epepsilon.example} with $p=2$ and $\epsilon=0$,
\begin{equation}
 \beta_{E_{2, 0}, T_1}= \sup_{\tilde t\ne 0} \sup_{|t|\le \pi/2}\Big\|\frac{(\tilde t\cos t, \tilde t\sin t)}
 {(\cos^2 t+\sin^2 t)^{1/2} |\tilde t|}-\frac{(\tilde t,0)}{|\tilde t|}\Big\|_2=\sqrt{2}
\end{equation}
where $T_1\tilde t=(\tilde t, 0), \tilde t\in \RR$.
}
\end{rem}

Define
 $$\theta_{F,T}=\sup_{u\in {\mathbf H}_1, v\ne 0}
\arccos \Big(\frac{\langle F'(u) v, Tv\rangle}{\|F'(u)v\| \|Tv\|}\Big),$$
the maximal angle between vectors
 $F'(u)v$ and $Tv$ in the Hilbert space ${\mathbf H}_2$.
  Then
 $$\beta_{F, T}
 =2\sin \frac{\theta_{F,T}}{2}.$$
 So the geometric condition \eqref{frame.sufficientcondition.hilbert.eq1}
 can be interpreted as that the angles between
 $F'(u)v$ and $Tv$ are less than or equal to   $\theta_{F,T}\in [0, \pi/2)$ for all $u, v\in {\mathbf H}_1$.
The above equivalence between
the geometric condition \eqref{frame.sufficientcondition.hilbert.eq1} and the angle condition $\theta_{F, T}<\pi/2$, together with
 \eqref{localstability.thm.eq1} and \eqref{Tstable.condition}, implies the existence of positive constants $A_1, B_1$ such that
\begin{equation}\label{positivity.condition}
A_1 \|Tv\|^2 \le  \langle F'(u)v, Tv\rangle\le  B_1 \|Tv\|^2, \ u,v\in {\mathbf H}_1.
\end{equation}
The converse can be proved to be true too. Thus $\beta_{F, T}<\sqrt{2}$ if and only if $S:=T^*F$ is strictly monotonic.
Here  a bounded map $S$ on a Hilbert space ${\mathbf H}$ is said to be {\em strictly monotonic} \cite{nonlinearbook} if
 there exist positive constants $m$ and $M$ such that
  \begin{equation*}
  m \|u-v\|^2\le \langle u-v, S(u)-S(v)\rangle \le M \|u-v\|^2 \ {\rm for \ all} \ u,v\in {\mathbf H}.
  \end{equation*}
  As an application of the above equivalence,  
  Theorem
\ref{frame.sufficientcondition.hilbert} can  be reformulated as follows.

\begin{thm} \label{frame.sufficientcondition.hilberts2}
Let  ${\mathbf H}_1$ and ${\mathbf H}_2$ be Hilbert spaces, and
let
$F:{\mathbf H}_1\to {\mathbf H}_2$
 be a continuously differentiable map with its derivative   having the uniform stability property
\eqref {localstability.thm.eq1}.
If there exists a linear operator $T\in {\mathcal B}({\mathbf H}_1, {\mathbf H}_2)$ satisfying \eqref{Tstable.condition} and
\eqref{positivity.condition},
then $F$ is a bi-Lipschitz map.
\end{thm}

From Theorem \ref{frame.sufficientcondition.hilberts2} we obtain the following result similar to
 the one in  Corollary \ref{linearapproximationcondition.cor}.

\begin{cor} \label{linearapproximationcondition.cor2}
Let  ${\mathbf H}_1, {\mathbf H}_2$
and
$F$ be as in Theorem \ref{frame.sufficientcondition.hilberts2}.
If there exists a  bounded below linear operator $T\in {\mathcal B}({\mathbf H}_1, {\mathbf H}_2)$
with $\delta_{F,T}<\sqrt{2}-1$,
then $F$ is a bi-Lipschitz map.
\end{cor}

Given a differentiable map $F$, it is quite technical in general to  construct linear  operator $T$ satisfying \eqref{Tstable.condition} and
\eqref{betaft.condition} in Banach space setting (respectively \eqref{Tstable.condition} and \eqref{positivity.condition} in Hilbert space setting). A conventional selection is that $T=F'(x_0)$ for some $x_0\in {\mathbf B}_1$, but such a selection is not always
favorable.
Let $\Phi=(\phi_\lambda)_{\lambda\in \Lambda}$
be impulse response vector with its entry $\phi_\lambda$ being the impulse response
of the signal generating device at the innovation position $\lambda\in \Lambda$, and
$\Psi=(\psi_\gamma)_{\gamma\in \Gamma}$
be sampling  functional vector with entry $\psi_\gamma$  reflecting the characteristics of the acquisition device
at the sampling position $\gamma\in \Gamma$.
In order to consider bi-Lipschitz property of the nonlinear sampling map
\begin{equation*}
S_{f, \Phi, \Psi}: \ell^2(\Lambda)\ni x \longmapsto x^T\Phi \overset{\rm companding } \longmapsto
f(x^T\Phi)
\overset{\rm sampling} \longmapsto
\langle f(x^T\Phi), \Psi\rangle \in \ell^2(\Gamma)
\end{equation*}
related to instantaneous
companding $h(t)\longmapsto f(h(t))$,  a linear operator
 $$T:= A_{\Phi, \Phi} (A_{\Phi, \Psi} (A_{\Psi, \Psi})^{-1} A_{\Psi, \Phi})^{-1} A_{\Phi, \Psi} (A_{\Psi, \Psi})^{-1}$$
satisfying \eqref{Tstable.condition} and
\eqref{positivity.condition} is implicitly introduced in \cite{sunaicm13},
 where $$A_{\Phi, \Psi}=(\langle \phi_\lambda, \psi_\gamma\rangle)_{\lambda\in \Lambda, \gamma\in \Gamma}$$
is the inter-correction matrix between $\Phi$ and $\Psi$.

\section{Sparse approximation triple} 
\label{sat.appendix}

In this appendix, we prove Propositions   \ref{greedy.prop} and
\ref{aA.prop}, and conclude it with a remark on the greedy algorithm \eqref{greedyalgorithm.m}.

\begin{proof}[Proof of Proposition \ref{greedy.prop}]
The convergence of $x_{\mathbf A, \mathbf M}^k, k\ge 0$, follows from
$$\sum_{k=0}^K\|x_{\mathbf A, \mathbf M}^{k+1}-x_{\mathbf A, \mathbf M}^k\|_{\mathbf M}=
\|x\|_{\mathbf M}-\|x-x_{\mathbf A, \mathbf M}^{K+1}\|_{\mathbf M}\le \|x\|_{\mathbf M}, \ K\ge 0,$$
 by the norm splitting property \eqref{optimization.thm1.eq3}.
Denote by $x_{\mathbf A, \mathbf M}^\infty\in {\mathbf M}$ the limit  of $x_{\mathbf A, \mathbf M}^{k}, k\ge 0$. Then
the limit
$x_{\mathbf A, \mathbf M}^\infty$ satisfies the following consistency condition:
\begin{equation}\label{consistentcondition}
\langle x_{\mathbf A, \mathbf M}^\infty, y\rangle=\langle x, y\rangle
\end{equation}
 for all $y\in {\mathbf A}$.
The above consistency condition holds as
 $0={\rm argmin}_{\hat x\in {\mathbf A}} \|x-x_{\mathbf A, \mathbf M}^\infty-\hat x\|_{\mathbf M}$, which together with
 the norm-splitting property \eqref{bestapproximation.hm} in ${\mathbf H}_1$
 implies that the projection of $x-x_{\mathbf A, \mathbf M}^\infty$ onto ${\mathbf A}_i$ are zero for all $i\in I$.
 From the consistency condition \eqref{consistentcondition}, we conclude that
 \eqref{consistentcondition} hold for all $y\in k{\mathbf A}, k\ge 0$, and hence for
 all $y$ in
 the closure of  $\cup_{k\ge 0} k{\mathbf A}$.
 This  together with the sparse density property
 of the sparse approximation triple
 $({\bf A}, {\bf M}, {\bf H}_1)$
 proves the convergence of
 $x_{\mathbf A, \mathbf M}^k, k\ge 0$, to $x\in {\mathbf M}$.
\end{proof}

\begin{proof}[Proof of Proposition \ref{aA.prop}]
Take $0\ne x\in {\mathbf M}$ and let $x_{\mathbf A, \mathbf M}^k, k\ge 0$,
be as in the greedy algorithm \eqref{greedyalgorithm.m}.
Write $u_k=x_{\mathbf A, \mathbf M}^{k+1}-x_{\mathbf A, \mathbf M}^{k}, k\ge 0$.
Thus
$$u_k={\rm argmin}_{\hat x\in {\mathbf A}} \|x-x_{{\mathbf A},{\mathbf M}}^k-\hat x\|_{\mathbf M}=
{\rm argmin}_{\hat x\in {\mathbf A}} \|(x-x_{{\mathbf A},{\mathbf M}}^{k-1})-u_{k-1}-\hat x\|_{\mathbf M}
\in {\mathbf A},$$
 and
$x-x_{\mathbf A, \mathbf M}=\sum_{k\ge 1} u_k$ by Proposition \ref{greedy.prop}.
This together with \eqref{optimization.thm1.eq3} and \eqref{secondapproximationestimate} implies that
\begin{equation*}
\|x-x_{\mathbf A, \mathbf M}\|_{{\mathbf H}_1}\le \sum_{k\ge 1} \|u_k\|_{{\mathbf H}_1}\le
\sqrt{a_{\mathbf A}} \sum_{k\ge 1} \|u_{k-1}\|_{\mathbf M}= \sqrt{a_{\mathbf A}}\|x\|_{\mathbf M}.
\end{equation*}
This completes the proof.
\end{proof}

Given a sparse approximation triple $({\bf A}, {\bf M}, {\bf H}_1)$, we say that
  $x\in {\mathbf M}$ is {\em compressible}  (\cite{bdieee11, 
  emieee09,  foucartbook, rauhut15}) if
 $\{\sigma_{k{\mathbf A}, {\bf M}}(x)\}_{k=1}^\infty$ having rapid decay,
such as
 $$\sigma_{k{\mathbf A}, {\bf M}}(x)\le C k^{-\alpha}\ {\rm for \ some} \ C,\ \alpha>0,$$
where
 $\sigma_{k{\mathbf A}, {\mathbf M}}(x)$
is  the best approximation error of $x$ from $k{\mathbf A}$,
\begin{equation}\label{sigmaam.def}
\sigma_{k{\mathbf A}, {\mathbf M}}(x)  
:=\inf_{\hat x\in k{\mathbf A}}\|\hat x-x\|_{\mathbf M}, \ k\ge 1.
\end{equation}

For the sequence $x_{\mathbf A, \mathbf M}^k, k\ge 0$, in the greedy algorithm
\eqref{greedyalgorithm.m},
 we have
\begin{equation}\label{greedyalgorithm.errorm}
\| x_{\mathbf A, \mathbf M}^k-x\|_{\mathbf M}\ge \sigma_{k{\mathbf A}, \mathbf M}(x),  \end{equation}
as $ x_{\mathbf A, \mathbf M}^k\in k{\mathbf A}, k\ge 1.$
The  above inequality  becomes an  equality in the classical sparse recovery setting.
We do not know  whether and when the greedy algorithm \eqref{greedyalgorithm.m} is suboptimal,
i.e., there exists a  positive constant $C$ such that
\begin{equation}
\| x_{\mathbf A, \mathbf M}^k-x\|_{\mathbf M} \le C \sigma_{k{\mathbf A}, \mathbf M}(x),\  x\in M,
\end{equation}
even for compressible signals.
The reader may refer to  \cite{greedybook} for the study of various greedy algorithms.

\end{appendix}

\bibliographystyle{siam}

\end{document}